\documentclass[conference]{IEEEtran}

\usepackage{amssymb,amsmath}
\usepackage{amsthm}

\newtheorem{definition}{Definition}

\usepackage{pifont}
\usepackage{hyperref}

\usepackage{xcolor,hyperref}
\usepackage{paralist}
\usepackage{url}
\usepackage{pgfplots}
\usepackage{xspace}
\usepackage[normalem]{ulem}

\usepackage{wrapfig}
\usepackage{tikz}
\usetikzlibrary{arrows,decorations,shapes,shadows,calc}

\newcommand{\rronly}[1]{#1} 
\newcommand{\pponly}[1]{}

\usepackage[linesnumbered,ruled,vlined]{algorithm2e}
\SetKw{KwLet}{let}
\SetKw{KwGlobal}{global}
\SetKw{KwAbort}{abort}
\SetKwFor{Forall}{for all}{do}{end}
\SetKwFor{Func}{function}{}{}
\SetKw{KwChoose}{choose}

\usepackage{listings}
\newcommand{\true}{\mathit{true}}
\newcommand{\false}{\mathit{false}}

\renewcommand{\vec}[1]{{\boldsymbol{#1}}}

\newcommand{\overtilde}[1]{{\overset{\sim}{#1}}}
\newcommand{\eqdef}{\equiv}

\newcommand{\templ}{\mathcal{T}}

\newcommand{\x}{{x}} 
\newcommand{\xv}{{\chi}} 

\newcommand{\vx}{{\vec{\x}}}
\newcommand{\vxp}{{\vec{\x}'}}
\newcommand{\vxpp}{{\vec{\x}''}}

\newcommand{\vxv}{{\vec{\xv}}}
\newcommand{\vxpv}{{\vec{\xv}'}}
\newcommand{\vxin}{{\vx^{in}}}
\newcommand{\vxinv}{{\vxv^{in}}}
\newcommand{\vxout}{{\vx^{out}}}
\newcommand{\vxargin}{{\vx^{p\_in}}}
\newcommand{\vxargout}{{\vx^{p\_out}}}
\newcommand{\rankparam}{\ell}

\newcommand{\rankparamvec}{L}
\newcommand{\rankparamvecv}{\Lambda}
\newcommand{\invparam}{d}
\newcommand{\invparamv}{\delta}
\newcommand{\precondparam}{e}

\newcommand{\rank}{\mathit{R}}
\newcommand{\rrank}{RR}
\newcommand{\lexrank}{LR}
\newcommand{\rranktempl}{\mathcal{RR}}
\newcommand{\lexranktempl}{\mathcal{LR}}

\newcommand{\precondtempl}{\mathcal{P}}

\newcommand{\trans}{\mathit{Trans}}
\newcommand{\inv}{\mathit{Inv}}
\newcommand{\invs}{\mathit{Invs}}

\newcommand{\init}{\mathit{Init}}
\newcommand{\fout}{\mathit{Out}}

\newcommand{\precond}{\mathit{Precond}}
\newcommand{\preconditions}{\mathit{Preconds}}
\newcommand{\summary}{\mathit{Sum}}
\newcommand{\summaries}{\mathit{Sums}}
\newcommand{\callctx}{\mathit{CallCtx}}
\newcommand{\assertions}{\mathit{Assertions}}
\newcommand{\assumptions}{\mathit{Assumptions}}
\newcommand{\termconditions}{\mathit{termConds}}

\newcommand{\comptermarg}{\mathit{compTermArg}}
\newcommand{\compprecondterm}{\mathit{compPrecondTerm}}
\newcommand{\compinv}{\mathit{compInv}}

\newcommand{\compinvsummary}{\mathit{compInvSum}}
\newcommand{\compcallctx}{\mathit{compCallCtx}}
\newcommand{\compprecond}{\mathit{compNecPrecond}}
\newcommand{\analyzeforward}{\mathit{analyzeForward}}
\newcommand{\analyzebackward}{\mathit{analyzeBackward}}

\newcommand{\ssadef}{\doteq}

\renewcommand{\function}{{procedure}}
\newcommand{\Function}{{Procedure}}

\newcommand*\rot{\rotatebox{90}}

\newcommand{\toolname}{2LS\xspace}

\newtheorem{theorem}{Theorem}
\newtheorem{lemma}{Lemma}
\newtheorem{proposition}{Proposition}
\newtheorem{remark}{Remark}

\renewcommand{\paragraph}[1]{\smallskip\noindent\textbf{#1}\quad}

\begin{document}

\title{Synthesising Interprocedural Bit-Precise Termination Proofs\rronly{ (extended version)}\thanks{The research leading to these results has received funding
    from the ARTEMIS Joint Undertaking under grant
    agreement number 295311 \href{http://vetess.eu/}{``VeTeSS''},
    and ERC project~280053 (CPROVER).}}

\author{
\IEEEauthorblockN{Hong-Yi Chen,  Cristina David,  Daniel Kroening, Peter Schrammel and Bj\"orn Wachter}
\IEEEauthorblockA{
Department of Computer Science, University of Oxford, first.lastname@cs.ox.ac.uk}
}

\maketitle

\begin{abstract}
Proving program termination is key to guaranteeing
absence of undesirable behaviour, such as hanging programs
and even security vulnerabilities such as denial-of-service attacks.
To make termination checks scale to large systems,
interprocedural termination analysis seems essential,
which is a largely unexplored area of research in termination analysis,
where most effort has focussed on difficult single-procedure problems.
We present a modular termination analysis for C programs
using template-based interprocedural summarisation.
Our analysis combines a context-sensitive, over-approximating forward
analysis with the inference of under-approximating preconditions
for termination.
Bit-precise termination arguments are synthesised over lexicographic
linear ranking function templates.
Our experimental results show that our tool \toolname outperforms
state-of-the-art alternatives, and demonstrate the clear
advantage of interprocedural reasoning over monolithic analysis in
terms of efficiency, while retaining comparable precision.
\end{abstract}



\section{Introduction}

Termination bugs can compromise safety-critical software systems by making them irresponsive,
e.g., termination bugs can be exploited in denial-of-service attacks~\cite{CVE}.
Termination guarantees are therefore instrumental for software reliability.
Termination provers,
static analysis tools that aim to construct a termination proof for a given input program,
have made tremendous progress.
They enable automatic proofs for complex loops that
may require linear lexicographic (e.g.~\cite{BG13,LH14}) or non-linear
termination arguments (e.g.~\cite{BMS05b}) in a completely automatic way.
However, there remain major practical challenges in analysing real-world code.

First of all, as observed by~\cite{FKS12}, most approaches in the
literature are specialised to linear arithmetic over unbounded
mathematical integers.  Although, unbounded arithmetic may reflect the
intuitively-expected program behaviour, the program actually executes
over bounded machine integers.  The semantics of C allows unsigned
integers to wrap around when they over/underflow.  Hence, arithmetic
on $k$-bit-wide unsigned integers must be performed modulo-$2^k$.
According to the C standards, over/underflows of signed integers are
undefined behaviour, but practically also wrap around on most
architectures.  Thus, accurate termination analysis requires a
\mbox{\emph{bit-precise}} analysis of program semantics.  Tools must
be configurable with architectural specifications such as the width of
data types and endianness.
The following examples illustrate that termination behaviour on
machine integers can be completely different than on mathematical
integers.  For example, the following code:
\begin{lstlisting}[numbers=none]
void foo1(unsigned n) { for(unsigned x=0; x<=n; x++); }
\end{lstlisting}
does terminate with mathematical integers, but does \emph{not} terminate
with machine integers if \texttt{n} equals the largest unsigned integer.
On the other hand, the following code:
\begin{lstlisting}[numbers=none]
void foo2(unsigned x) { while(x>=10) x++; }
\end{lstlisting}
does not terminate with mathematical integers, but terminates with
machine integers because unsigned machine integers wrap around.

A second challenge is to make termination analysis scale to larger programs.
The yearly Software Verification Competition (SV-COMP)~\cite{DBLP:conf/tacas/Beyer15}
includes a division in termination analysis, which reflects a representative picture of the state-of-the-art. 
The SV-COMP'15 termination benchmarks contain challenging termination problems on smaller programs
with at most 453 instructions \hbox{(average 53)}, contained at most 7 functions \hbox{(average 3)}, and 4 loops \hbox{(average 1)}.

In this paper, we present a technique that we have successfully
run on programs that are one magnitude larger, containing up to 5000 instructions.
Larger instances require different algorithmic
techniques to scale, e.g., modular interprocedural analysis rather than monolithic analysis.
This poses several conceptual and practical challenges that do not
arise in monolithic termination analysers.
For example, when proving termination of a program, a possible approach is to try
to prove that all {\function}s in the program terminate \emph{universally},
i.e., in any possible calling context.
However, this criterion is too optimistic, as termination of individual
{\function}s often depends on the calling context,
i.e., {\function}s terminate \emph{conditionally} only in specific calling contexts.

Hence, an interprocedural analysis strategy is to verify universal program termination
in a top-down manner by proving termination of each {\function}
relative to its \emph{calling contexts}, and propagating upwards which
calling contexts guarantee termination of the {\function}.
It is too difficult to determine these contexts precisely; analysers thus compute preconditions for
termination.  A \emph{sufficient precondition} identifies those pre-states in
which the {\function} will definitely terminate, and is thus suitable for
proving termination.  By contrast, a \emph{necessary precondition}
identifies the pre-states in which the {\function} may terminate.  Its
negation are those states in which the {\function} will not terminate, which is
useful for proving nontermination.  

In this paper we focus on the computation of sufficient preconditions. 
Preconditions enable information reuse, and thus scalability, as it is
frequently possible to avoid repeated analysis of parts of the code base, e.g. 
libraries whose {\function}s are called multiple times or did not undergo
modifications between successive analysis runs.

\paragraph{Contributions:}
\begin{compactenum}
\item We propose an algorithm for \emph{interprocedural termination analysis}.
  The approach is based on a template-based static analysis using SAT
  solving.  It combines context-sensitive, summary-based
  interprocedural analysis with the inference of preconditions for
  termination based on template abstractions.
  We focus on non-recursive programs, which cover a large portion of
  software written, especially in domains such as embedded systems.
\item We provide an implementation of the approach in \toolname, a
  static analysis tool for C programs.  Our instantiation of the
  algorithm uses template polyhedra and lexicographic, linear ranking
  functions templates.  The analysis is bit-precise and purely relies
  on SAT-solving techniques.

\item We report the results of an experimental evaluation on 597
  procedural SV-COMP benchmarks with in total 1.6 million lines of code
  that demonstrates the scalability and applicability of the approach
  to programs with thousands of lines of code.
\end{compactenum}

%

\section{Preliminaries}\label{sec:prelim}

In this section, we introduce basic notions of
interprocedural and termination analysis.

\paragraph{Program model and notation.}
We assume that programs are given in terms of acyclic\footnote{
We consider non-recursive programs with multiple procedures.} 
call graphs,
where individual {\function}s $f$ are given in terms of
symbolic input/output transition systems.
Formally, the input/output transition system of a {\function} $f$ 
is a triple $(\init_f,\trans_f,\fout_f)$, where $\trans_f(\vx,\vxp)$ is the
transition relation; the input relation $\init_f(\vxin, \vx)$ defines the
initial states of the transition system and
relates it to the inputs $\vxin$; the output relation
$\fout_f(\vx,\vxout)$ connects the transition system to the outputs
$\vxout$ of the {\function}.
Inputs are {\function} parameters, global variables, and memory
objects that are read by $f$.  Outputs are return values, and
potential side effects such as global variables and memory objects
written by $f$. Internal states $\vx$ are commonly the values of variables at
the loop heads in $f$.

These relations are given \emph{as first-order logic formulae}
resulting from the logical encoding of the program semantics.
Fig.~\ref{fig:encoding1} shows the encoding of the two {\function}s
in Fig.~\ref{fig:example1} into such formulae.%
\footnote{\emph{c?a:b} is the conditional operator, which returns $a$ if $c$ evaluates to $\true$, and $b$ otherwise.}
The inputs $\vxin$ of $f$ are $(z)$ and the outputs $\vxout$ consist
of the return value denoted $(r_f)$.  The transition relation of $h$
encodes the loop over the internal state variables $(x,y)^T$.  
We may need to introduce Boolean variables $g$ to model the control flow,
as shown in $f$.
Multiple and nested loops can be similarly encoded in $\trans$.

Note that we view these formulae as predicates, e.g.\ $\trans(\vx,
\vxp)$, with given parameters $\vx, \vxp$, and mean the substitution
$\trans[\vec{a}/\vec{x},\vec{b}/\vec{x}']$ when we write
$\trans(\vec{a}, \vec{b})$.
Moreover, we write $\vx$ and $x$ with the understanding that the
former is a vector, whereas the latter is a scalar.

Each call to a {\function} $h$ at call site $i$ in a {\function}~$f$ is
modeled by a \emph{placeholder predicate}
\hbox{$h_i(\vxargin_i,\vxargout_i)$} occurring in the formula $\trans_f$ for $f$.
The placeholder predicate ranges over intermediate variables
representing its actual input and output parameters $\vxargin_i$ and
$\vxargout_i$, respectively.  Placeholder predicates 
evaluate to $\true$, which corresponds to havocking {\function} calls.
In {\function} $f$ in Fig.~\ref{fig:encoding1}, the placeholder for the
{\function} call to $h$ is $h_0((z),(w_1))$ with the actual input and
output parameters $z$ and $w_1$, respectively.

A full description of the program encoding 
is given in \rronly{Appendix~\ref{sec:appendix}}

\paragraph{Basic concepts.}
Moving on to interprocedural analysis, we introduce formal notation 
for the basic concepts below: 
\begin{definition}[Invariants, Summaries, Calling Contexts] \label{def:inv}
For a {\function} given by 
\hbox{$(\init,\trans, \linebreak \fout)$}
we define: 
\begin{itemize}
  \item An
  \emph{invariant} is a predicate $\inv$ such that:
  \[
  \begin{array}{rlr}
  \forall \vxin, \vx, \vxp: & \init(\vxin,\vx) \Longrightarrow \inv(\vx) 
  \\
  \wedge &\inv(\vx)\wedge\trans(\vx,\vxp)\Longrightarrow \inv(\vxp) 
  \end{array}
  \]
  \item Given an invariant $\inv$, a \emph{summary} is a
    predicate $\summary$ such that:
  \[
  \begin{array}{rl}
  \forall \vxin,\vx, \vxp,\vxout:
  & \init(\vxin,\vx) \wedge \inv(\vxp) \wedge \fout(\vxp,\vxout) \\
  & \Longrightarrow \summary(\vxin,\vxout)
  \end{array}
  \]
  
  \item
  Given an invariant $\inv$, the \emph{calling context} for a {\function} call $h$ at call site $i$ in the given {\function} is a predicate $\callctx_{h_i}$ such that 
  \[
  \begin{array}{l}
  \forall \vx,\vxp,\vxargin_i,\vxargout_i: \\
  \quad \inv(\vx) \wedge \trans(\vx,\vxp) \Longrightarrow \callctx_{h_i}(\vxargin_i,\vxargout_i)
  \end{array}
  \]
\end{itemize}
\end{definition}
These concepts have the following roles: 
Invariants abstract the behaviour of loops. 
Summaries abstract the behaviour of called {\function}s;
they are used to strengthen the placeholder predicates.
Calling contexts abstract the caller's behaviour w.r.t.\ the
{\function} being called. When analysing the callee, the calling
contexts are used to constrain its inputs and outputs.
In Sec.~\ref{sec:overview} we will illustrate these notions on the
program in Fig.~\ref{fig:example1}.

\begin{figure}[t]
\begin{tabular}{l@{\hspace{3em}}l}
\begin{lstlisting}
unsigned f(unsigned z) {
  unsigned w=0;
  if(z>0) w=h(z);
  return w;
}
\end{lstlisting}
&
\begin{lstlisting}
unsigned h(unsigned y) {
  unsigned x;
  for(x=0; x<10; x+=y);
  return x;
}
\end{lstlisting}
\end{tabular}
\caption{\label{fig:example1}
Example.
}
\vspace{-.25cm}
\end{figure}

\begin{figure}[t]
\footnotesize
\begin{tabular}{@{\hspace*{-0.5em}}r@{\,}l}
$\init_f((z),(w,z,g)^T) \equiv$ & $(w_0{=}0 \wedge z'{=}z \wedge g)$ \\
$\trans_f((w,z,g)^T,(w',z',g')^T) \equiv$ & $(g \wedge h_0((z),(w_1)) \wedge$\\
&                          $w {=} (z{>}0 ? w_1{:}w_0) \wedge \neg g')$ \\
$\fout_f((w,z,g)^T,(r_f)) \equiv$ & $(r_f{=}w)$ \\[0.5ex]
\hline
\\[-1.5ex]
$\init_h((y),(x,y')^T) \equiv$ & $(x{=}0 \wedge y'{=}y)$ \\
$\trans_h((x,y)^T,(x',y')^T) \equiv$ & $(x'{=}x{+}y \wedge x{<}10 \wedge y{=}y')$  \\
$\fout_h((x,y)^T,(r_h)) \equiv$ & $(r_h{=}x \wedge \neg(x{<}10))$
\\
\end{tabular}
\caption{\label{fig:encoding1}
Encoding of Example~\ref{fig:example1}.
}
\vspace{-.25cm}
\end{figure}

Since we want to reason about termination,
we need the notions of ranking functions and preconditions
for termination.

\begin{definition}[Ranking function] \label{def:ranking_function}
A \emph{ranking function} for a {\function} $(\init,\trans,\fout)$ with
invariant $\inv$ is a function $r$
from the set of program states 
to a well-founded domain 
such that
$
\forall \vx,\vxp:
\inv(\vx) \wedge \trans(\vx,\vxp)
\Longrightarrow r(\vx) > r(\vxp).
$
\end{definition}

We denote by $\rrank(\vx,\vxp)$ a set of constraints that
guarantee that $r$ is a ranking function.  
The existence of a ranking function for a {\function} guarantees its
\emph{universal} termination.

The weakest termination precondition for a {\function} describes the
inputs for which it terminates.
If it is $\true$, the {\function} terminates universally; if it is
$\false$, then it does not terminate for any input.
Since the weakest precondition is intractable to compute or even
uncomputable, we under-approximate the precondition.  A
\emph{sufficient precondition} for termination guarantees that the program
terminates for all $\vxin$ that satisfy it.

\begin{definition}[Precondition for termination]\label{def:precond}
Given a {\function} $(\init,\trans,\fout)$, a
sufficient \emph{precondition for termination} is a predicate
$\precond$ such that
\[ \begin{array}{rl}
\multicolumn{2}{l}{\exists \rrank,\inv: \forall \vxin,\vx,\vxp:} \\
& \precond(\vxin) \wedge \init(\vxin,\vx) \Longrightarrow \inv(\vx) \\
\wedge & \inv(\vx) \wedge \trans(\vx,\vxp) \Longrightarrow \inv(\vxp)
 \wedge \rrank(\vx,\vxp)
\end{array} \]
\end{definition}
Note that $\false$ is always a trivial model for $\precond$, but not a
very useful one.

\section{Overview of the Approach}\label{sec:overview}

In this section, we introduce the architecture of 
our interprocedural termination analysis.
Our analysis combines, in a non-trivial synergistic way,
the inference of invariants, summaries, calling contexts,
termination arguments, and preconditions,
which have a concise characterisation in second-order
logic (see Definitions~\ref{def:inv}, and~\ref{def:precond}).
At the lowest level our approach relies on a solver backend for second-order problems,
which is described in Sec.~\ref{sec:sa}.

To see how the different analysis components fit together,
we now go through the pseudo-code of our termination analyser (Algorithm~\ref{alg:interproc}). Function $\mathit{analyze}$ is given the
entry {\function} $f_{\mathit{entry}}$ of the program as argument and
proceeds in two analysis phases.

Phase one is an \emph{over-approximate} forward analysis,
given in subroutine $\analyzeforward$, which 
recursively descends into the call graph from the entry point $f_{\mathit{entry}}$.
Subroutine $\analyzeforward$ infers for each {\function} call in $f$ an
over-approximating calling context $\callctx^o$,
using {\function} summaries and other previously-computed information.  
Before analyzing a callee, 
the analysis checks if the callee has 
already been analysed and, whether the stored summary 
can be re-used, i.e., if it is 
compatible with the new calling context $\callctx^o$.
Finally, once summaries for all callees are available, the analysis infers loop
invariants and a summary for $f$ itself, which are stored for later re-use 
by means of a join operator.

The second phase is an \emph{under-approximate} backward analysis,
subroutine $\analyzebackward$, which infers termination preconditions.  
Again, we recursively descend into the call graph.
Analogous to the forward analysis, we infer for each {\function} call
in $f$ an under-approximating calling context $\callctx^u$ (using
under-approximate summaries, as described in
Sec.~\ref{sec:interproc}), and recurses only if necessary
(Line~\ref{line:bw_recurse}).
Finally, we compute the under-approximating precondition for termination
(Line~\ref{line:bw_precondterm}).
This precondition is inferred w.r.t.~the termination conditions that
have been collected: the backward calling context
(Line~\ref{line:bw_cond1}), the preconditions for termination of the
callees (Line~\ref{line:bw_cond2}), and the termination arguments for
$f$ itself (see Sec.~\ref{sec:interproc}).
Note that superscripts $o$ and $u$
in predicate symbols indicate over- and underapproximation, respectively.

\SetAlFnt{\small}
\begin{algorithm}[t]
\KwGlobal $\summaries^o,\invs^o,\preconditions^u$\;
\Func{$\analyzeforward(f,\callctx^o_f)$}{
  \ForEach {{\function} call $h$ in $f$}
  {
    $\callctx^o_h = \compcallctx^o(f,\callctx^o_f,h)$\;
    \If{$\mathit{needToReAnalyze}^o(h,\callctx^o_h)$\label{line:needtoreanalze1}}
    {
     $\analyzeforward(h,\callctx^o_h)$\;
    }
  }
  $\mathit{join}^o((\summaries^o[f],\invs^o[f]),\compinvsummary^o(f,\callctx^o_f))$\label{line:fw_invsum}
}
\Func{$\analyzebackward(f,\callctx^u_f)$}{
  $\termconditions = \callctx^u_f$\label{line:bw_cond1}\;
  \ForEach {{\function} call $h$ in $f$}
  {
    $\callctx^u_h = \compcallctx^u(f,\callctx^u_f,h)$\;
    \If{$\mathit{needToReAnalyze}^u(h,\callctx^u_h)$\label{line:bw_recurse}}
    {
      $\analyzebackward(h,\callctx^u_h)$\;
    }
    $\termconditions \gets \termconditions \wedge \preconditions^u[h]$\label{line:bw_cond2}\;
  }
  $\begin{array}{@{}l@{}l}\mathit{join}^u(&\preconditions^u[f],\\
      &\compprecondterm(f,\invs^o[f],\termconditions);
      \end{array}$\label{line:bw_precondterm}
}
\Func{$\mathit{analyze}(f_\mathit{entry})$}{
  $\analyzeforward(f_\mathit{entry},\true)$\;
  $\analyzebackward(f_\mathit{entry},\true)$\;
  \Return $\preconditions^u[f_\mathit{entry}]$\;
}
\caption{\label{alg:interproc}
$\mathit{analyze}$}
\end{algorithm}
%

\paragraph{Challenges.} Our algorithm uses 
over- and under-approximation in a novel, systematic way. 
In particular, we address the challenging problem of finding meaningful preconditions:
\begin{itemize}
\item  
The precondition Definition~\ref{def:precond} admits the trivial solution
$\false$ for $\precond$. How do we find a good candidate?
To this end, we ``bootstrap'' the process with a
candidate precondition: a single value of $\vxin$, for which we
compute a termination argument.
The key observation is that the resulting termination argument is typically
more general, i.e., it shows termination for many further entry states.
The more general precondition is then computed by precondition inference
w.r.t.~the termination argument.
\item 
A second challenge is to compute under-approximations. 
Obviously, the predicates in
the definitions in Sec.~\ref{sec:prelim} can be over-approximated
by using abstract domains such as intervals.
However, there are only few methods for under-approximating analysis.  In
this work, we use a method similar to \cite{CGL+08} to obtain
under-approximating preconditions w.r.t.~property~$p$: we infer an
over-approximating precondition w.r.t.~$\neg p$ and negate the result. 
In our case, $p$ is the termination condition $\termconditions$.
\end{itemize}

\paragraph{Example.}
We illustrate the algorithm on the simple example given as
Fig.~\ref{fig:example1} with the encoding in Fig.~\ref{fig:encoding1}.
\texttt{f} calls a {\function} \texttt{h}.  {\Function} \texttt{h}
terminates if and only if its argument \texttt{y} is non-zero, i.e.,
{\function} \texttt{f} only terminates conditionally.  The call of
\texttt{h} is guarded by the condition \texttt{z>0}, which guarantees
universal termination of {\function} \texttt{f}.

Let us assume that unsigned integers are 32 bits wide, and we use an
interval abstract domain for invariant, summary and precondition
inference, but the abstract domain with the elements
$\{\true,\false\}$ for computing calling contexts, i.e., we can prove
that calls are unreachable. We use $M:=2^{32}{-}1$.

Our algorithm proceeds as follows.  The first phase is
$\analyzeforward$, which starts from the entry {\function} \texttt{f}.
By descending into the call graph, we must compute an
over-approximating calling context $\callctx^o_{h}$ for {\function}
\texttt{h} for which no calling context has been computed before.
This calling context is $\true$.
Hence, we recursively analyse \texttt{h}.  Given that \texttt{h} does
not contain any {\function} calls, we compute the over-approximating
summary $\summary^o_h=(0{\leq} y{\leq} M \wedge 0 {\leq} r_h{\leq} M)$
and invariant $\inv^o_h=(0{\leq} x{\leq} M \wedge 0{\leq} y{\leq} M)$.
Now, this information can be used in order to compute
$\summary^o_f=(0{\leq} z{\leq} M \wedge 0 {\leq} r_f{\leq} M)$ and
invariant $\inv^o_f=\true$ for the entry {\function} \texttt{f}.

The backwards analysis starts again from the entry {\function} \texttt{f}. 
It computes an under-approximating calling context $\callctx^u_{h}$ for
{\function} \texttt{h}, which is $\true$, before descending into the call
graph.
It then computes an under-approximating precondition for termination
$\precond^u_{h} = (1{\leq} y{\leq} M)$ or, more precisely, an
under-approximating summary whose projection onto the input variables of
\texttt{h} is the precondition $\precond^u_{h}$.
By applying this summary at the call site of \texttt{h} in
\texttt{f}, we can now compute the precondition for termination
$\precond^u_{f} = (0{\leq} z{\leq} M)$ of \texttt{f}, which proves universal
termination of~\texttt{f}.

We illustrate the effect of the choice of the abstract domain on the
analysis of the example program.  Assume we replace the $\{\true,\false\}$
domain by the interval domain.  In this case, $\analyzeforward$ computes
$\callctx^o_{h}=(1{\leq} z{\leq} M\wedge 0{\leq} w_1{\leq} M)$.
The calling context is computed over the actual parameters $z$ and
$w_1$.  It is renamed to the formal parameters $y$ and $r_h$ (the
return value) when $\callctx^o_{h}$ is used for constraining the
pre/postconditions in the analysis of \texttt{h}.
Subsequently, $\analyzebackward$ computes the precondition
for termination of \texttt{h} using the union of all calling contexts
in the program. Since \texttt{h} terminates unconditionally
in these calling contexts, we trivially obtain $\precond^u_{h} =
(1{\leq} y{\leq} M)$, which in turn proves universal
termination of~\texttt{f}.

\section{Interprocedural Termination Analysis} \label{sec:interproc}

We can view Alg.~\ref{alg:interproc} 
as solving a series of formulae in second-order
predicate logic with existentially quantified predicates,
for which we are seeking satisfiability witnesses.%
\footnote{
To be precise, we are not only looking for witness predicates but
(good approximations of) weakest or strongest predicates. 
Finding such biased witnesses is a feature of our synthesis algorithms.}
In this section,  
we state the constraints we solve, including 
all the side constraints arising from the interprocedural analysis.
Note that this is not a formalisation exercise, but these are
precisely the formulae solved by our synthesis backend,
which is described in Section~\ref{sec:sa}.

\subsection{Universal Termination}\label{sec:univterm}

\SetAlFnt{\small}
\begin{algorithm}[t]
\KwGlobal $\summaries^o,\invs^o,\mathit{termStatus}$\;
\Func{$\analyzeforward(f,\callctx^o_f)$}{
  \ForEach {{\function} call $h$ in $f$}
  {
    $\callctx^o_h = \underline{\compcallctx^o}(f,\callctx^o_f,h)$\;
    \If{$\mathit{needToReAnalyze}^o(h,\callctx^o_h)$}
    {
     $\analyzeforward(h,\callctx^o_h)$\;
    }
  }
  $\mathit{join}^o((\summaries^o[f],\invs^o[f]),\underline{\compinvsummary^o}(f,\callctx^o_f))$
}
\Func{$\analyzebackward'(f)$}{
    $\mathit{termStatus}[f] = \underline{\comptermarg}(f)$\;
  \ForEach {{\function} call $h$ in $f$}
  {
   \If{$\mathit{needToReAnalyze}^u(h,\callctx^o_h)$} 
    {
      $\analyzebackward(h)$\;
      $\mathit{join}(\mathit{termStatus}[f],\mathit{termStatus}[h])$\label{line:join2}\;
    }
  }
}
\Func{$\mathit{analyze}(f_\mathit{entry})$}{
  $\analyzeforward(f_\mathit{entry},\true)$\;
  $\analyzebackward'(f_\mathit{entry})$\;
  \Return $\mathit{termStatus}[f_\mathit{entry}]$\;
}
\caption{\label{alg:interprocuniversal}
$\mathit{analyze}$ for universal termination}
\end{algorithm}
%

For didactical purposes, we start with a simplification of
Algorithm~\ref{alg:interproc} that is able to show universal
termination (see Algorithm~\ref{alg:interprocuniversal}).  This
variant reduces the backward analysis to a call to $\comptermarg$ and
propagating back the qualitative result obtained: \emph{terminating,
  potentially non-terminating}, or \emph{non-terminating}.

This section states
the constraints that are solved to
compute the outcome of the functions underlined in
Algorithm~\ref{alg:interprocuniversal} and establish its soundness:
\begin{compactitem}
\item $\compcallctx^o$ (Def.~\ref{prop:callctxo})
\item $\compinvsummary^o$ (Def.~\ref{prop:summaryo})
\item $\comptermarg$ (Lemma~\ref{prop:term1})\\
\end{compactitem}

\begin{definition}[$\compcallctx^o$]\label{prop:callctxo} 
  A forward calling context $\callctx^o_{h_i}$
  for $h_i$ in {\function} $f$ in calling context $\callctx^o_f$ is
  a satisfiability witness of the following formula:
$$
\begin{array}{r@{}r@{}l}
  \multicolumn{3}{l}{\exists \callctx^o_{h_i}, \inv^o_f: \forall \vxin,\vx,\vxp,\vxout,\vxargin_i,\vxargout_i:} \\
\multicolumn{3}{l}{~~\callctx^o_f(\vxin,\vxout) \wedge \summaries^o_f \rronly{ \wedge \assumptions_f(\vx)} \Longrightarrow} \\
& \big( & \init_f(\vxin,\vx) \Longrightarrow \inv^o_f(\vx) \big) \\
&~~\wedge\big( & \inv^o_f(\vx) \wedge 
\trans_f(\vx,\vxp) \\
&&\Longrightarrow \inv^o_f(\vxp) \wedge (g_{h_i} \Rightarrow\callctx^o_{h_i}(\vxargin_i,\vxargout_i)\big)
\end{array}
$$
\end{definition}
$\begin{array}{@{}rl}
\text{with }\summaries^o_f =
\bigwedge_{\text{calls }h_j\text{ in }f} &  g_{h_j} \Longrightarrow \\
&\summaries^o[h](\vxargin_j,\vxargout_j)
\end{array}
$\\ where $g_{h_j}$ is the guard condition of {\function} call $h_j$
in $f$ capturing the branch conditions from conditionals.
For example, $g_{h_0}$ of the {\function} call to \texttt{h} in \texttt{f} in
Fig.~\ref{fig:example1} is $z>0$. $\summaries^o[h]$ is the currently
available summary for \texttt{h} (cf.\ global variables in Alg.~\ref{alg:interproc}).
\rronly{Assumptions correspond to \texttt{assume()} statements in the code.}  
\begin{lemma}\label{lemma:callctxo}
$\callctx^o_{h_i}$ is over-approximating.
\end{lemma}
\begin{proof}
$\callctx^o_f$ when $f$ is the entry-point {\function} is $\true$; also, the
summaries $\summary^o_{h_j}$ are initially assumed to be $\true$,
i.e.~over-appro\-xi\-ma\-ting.
Hence, given that $\callctx^o_f$ and $\summaries^o_f$ are
over-approxi\-ma\-ting, $\callctx^o_{h_i}$ is over-approximating by 
the soundness of the synthesis (see Thm.~\ref{thm:synthsound} in Sec.~\ref{sec:sa}).
\end{proof}

\paragraph{Example.} 
Let us consider {\function} \texttt{f} in Fig.~\ref{fig:example1}.
\texttt{f} is the entry {\function}, hence we have
$\callctx^o_{f}((z),(r_{f})) = \true$ ($=(0 {\leq} z {\leq} M \wedge 0
       {\leq} r_{f} {\leq} M)$ with $M:=2^{32}{-}1$ when using the
       interval abstract domain for 32 bit integers).  Then, we
       instantiate Def.~\ref{prop:callctxo} (for {\function}
       \texttt{f}) to compute $\callctx^o_{h_0}$. We assume that we
       have not yet computed a summary for \texttt{h}, thus,
       $\summary_h$ is $\true$. Remember that the placeholder $h_0((z),(w_1))$ 
       evaluates to $\true$.
         \rronly{Notably, there are no
         assumptions in the code, meaning that $\assumptions_{f}(z) =
         true$.}

\medskip
\centerline{
$
\begin{array}{@{}r@{}r@{}l}
  \multicolumn{3}{l}{\exists \callctx^o_{h_0}, \inv^o_{f}: \forall z,w_1,w,w',z',g,g',r_{f} :} \\
\multicolumn{3}{l}{~~0 {\leq} z {\leq} M \wedge 0 {\leq} r_f {\leq} M \wedge (z{>}0 \Longrightarrow \true) \rronly{ \wedge \true} \Longrightarrow} \\
& \big( & w{=}0 \wedge z'{=}z \wedge g \Longrightarrow \inv^o_f((w,z,g)^T) \big) \\
&~~\wedge\big( & \inv^o_f((w,z,g)^T) \wedge \\
&&g \wedge h_0((z),(w_1)) \wedge w' {=} (z{>}0 ? w_1{:}w) \wedge z'{=}z\wedge \neg g' \\
&&\Longrightarrow \inv^o_f((w',z',g')^T) \wedge (z{>}0 \Rightarrow\callctx^o_{h_i}((z),(w_1))\big)
\end{array}
$
}

\smallskip
A solution is $\inv^o_{f} = \true$, and
$\callctx^o_{h_0}((z),(w_1)) =(1{\leq} z{\leq} M \wedge 0{\leq} w_1{\leq} M)$. 

\begin{definition}[$\compinvsummary^o$]\label{prop:summaryo} 
A forward summary $\summary^o_f$ and invariants $\inv^o_f$ for {\function} $f$ in calling context $\callctx^o_f$ are 
satisfiability witnesses of the following formula:

\medskip
\centerline{
$
\begin{array}{r@{}r@{}l}
\multicolumn{3}{l}{\exists \summary^o_f, \inv^o_f: \forall \vxin,\vx,\vxp,\vxpp,\vxout:} \\
\multicolumn{3}{l}{~~\callctx^o_f(\vxin,\vxout) \wedge \summaries^o_f \rronly{ \wedge \assumptions_f(\vx)} \Longrightarrow} \\
& \big( & \init_f(\vxin,\vx) \wedge 
\inv^o_f(\vxpp) \wedge
\fout_f(\vxpp,\vxout) \\
&&\Longrightarrow \inv^o_f(\vx) \wedge \summary^o_f(\vxin,\vxout) \big) \\
&~~\wedge\big( & \inv^o_f(\vx) \wedge 
\trans_f(\vx,\vxp) \Longrightarrow \inv^o_f(\vxp)\big)
\end{array}
$}
\end{definition}
\begin{lemma}\label{lemma:summaryo}
$\summary^o_f$ and $\inv^o_f$ are over-approximating.
\end{lemma}
\begin{proof}
By Lemma~\ref{lemma:callctxo}, $\callctx^o_f$ is over-approximating.
Also, the summaries $\summaries^o_f$ are initially assumed to be $\true$,
i.e. over-approxi\-ma\-ting.
Hence, given that $\callctx^o_f$ and $\summaries^o_f$ are
over-approxi\-ma\-ting, $\summary^o_f$ and $\inv^o_f$ are
over-approximating by the soundness of the synthesis
(Thm.~\ref{thm:synthsound}).
\end{proof}

\paragraph{Example.} 
Let us consider {\function} \texttt{h} in Fig.~\ref{fig:example1}.  We
have computed $\callctx^o_{h_0}((y),(r_h)) = (1{\leq} y{\leq} M \wedge
0{\leq} r_h{\leq}M)$ (with actual parameters renamed to formal ones).
Then, we need obtain witnesses $\inv^o_{h_0}$ and $\summary^o_{h_0}$
to the satifiability of the instantiation of Def.~\ref{prop:summaryo}
(for {\function} \texttt{h}) as given below.  

\medskip
\centerline{
$
\begin{array}{r@{}r@{}l}
\multicolumn{3}{l}{\exists \inv^o_{h_0}, \summary^o_{h_0}: \forall y,x,x',y',x'',y'',r_f:} \\
\multicolumn{3}{l}{~~1 {\leq} y {\leq} M \wedge
0{\leq} r_h{\leq}M \wedge \true \Longrightarrow}\\
& \big( & (x{=}0 \wedge y'{=}y) \wedge \inv^o_{h}((x'',y'')^T) \wedge 
(r_h{=}x'' \wedge \neg(x''{<}10)\\
&&    \Longrightarrow \inv^o_{h}((x,y')^T) \wedge \summary^o_{h}((y),(r_h))\big) \\
& \quad \wedge \big( &
\inv^o_{h}((x,y)^T) \wedge (x'{=}(x{+}y \wedge x{<}10) \wedge y{=}y') \\
&&\Longrightarrow \inv^o_{h}((x',y')^T)\big)
\end{array}
$}

A solution is 
$\inv^o_{h_0} =(0{\leq} x{\leq} M \wedge 1{\leq} y{\leq} M)$ and
$\summary^o_{h_0} =(1{\leq} y{\leq} M \wedge 10{\leq} r_h{\leq} M)$, for instance.

\begin{remark}\label{rem:ipfixpoint}
Since Def.~\ref{prop:callctxo} and Def.~\ref{prop:summaryo}
are interdependent, we can compute them iteratively 
until a fixed point is reached in order to improve the precision of
calling contexts, invariants and summaries.
However, for efficiency reasons, we perform only the first iteration
of this (greatest) fixed point computation.
\end{remark}

\begin{lemma}[$\comptermarg$]\label{prop:term1}
A {\function} $f$
with forward invariants $\inv^o_f$ terminates if
there is a termination argument $\rrank_f$:
$$
\begin{array}{ll@{}l}
\multicolumn{3}{l}{\exists \rrank_f: \forall \vx,\vxp:} \\
&& \inv^o_f(\vx) \wedge \trans_f(\vx,\vxp) \wedge \rronly{\\
&&} \summaries^o_f \wedge \rronly{\assumptions_f(\vx) \wedge} \assertions_f(\vx) \\
&& \Longrightarrow \rrank_f(\vx,\vxp)
\end{array}
$$
\end{lemma}
Assertions in this formula correspond to to \texttt{assert()}
statements in the code. They can be assumed to hold because
assertion-violating traces terminate.
Over-approximating forward information may lead to inclusion of
spurious non-terminating traces. For that reason, we might not find a
termination argument although the {\function} is terminating.  As
we essentially under-approximate the set of terminating {\function}s, we will 
not give false positives. 
Regarding the solving algorithm for this formula,
we refer to Sec.~\ref{sec:sa}.

\paragraph{Example.}
Let us consider function \texttt{h} in Fig.~\ref{fig:example1}.
Assume we have the invariant $0 {\leq} x {\leq} M \wedge 1 {\leq} y {\leq} M$. 
Thus, we have to solve
$$
\begin{array}{l}
\exists \rrank_h:
0 {\leq} x {\leq} M \wedge 1 {\leq} y {\leq} M \wedge 
x'{=}x{+}y \wedge x{<}10 \wedge y'{=}y \wedge\\
\qquad \qquad \true \wedge \true \Longrightarrow \rrank_h((x,y),(x',y'))
\end{array}
$$
When using a linear ranking function template $c_1\cdot x+c_2\cdot y$, 
we obtain as solution, for example, $\rrank_h = (-x{>}-x')$.

\smallskip
If there is no trace from {\function} entry to exit, then we can prove
non-termination, even when using over-appro\-xi\-mations:
\begin{lemma}[line~\ref{line:fw_invsum} of $\mathit{analyze}$%
]\label{prop:nonterm}
A {\function} $f$
in forward calling context $\callctx^o_f$,
and forward invariants $\inv^o_f$ never terminates if
its summary $\summary^o_f$ is $\false$.
\end{lemma}

Termination information is then propagated in the (acyclic) call graph
($\mathit{join}$ in line~\ref{line:join2} in
Algorithm~\ref{alg:interprocuniversal}):
\begin{proposition}\label{prop:ipterm}
A {\function} is declared
\begin{compactenum}
\item[(1)] \emph{non-terminating} if it is
non-terminating by Lemma~\ref{prop:nonterm}.
\item[(2)] \emph{terminating} if
\begin{compactenum}
\item[(a)] all its {\function} calls $h_i$ that
are potentially reachable (i.e.\ with $\callctx^o_{h_i}\neq \false$)
are declared terminating, and
\item[(b)] $f$ itself is terminating according to Lemma~\ref{prop:term1};
\end{compactenum}
\item[(3)]  \emph{potentially
    non-terminating}, otherwise.
\end{compactenum}
\end{proposition}
Our implementation is more efficient than
Algorithm~\ref{alg:interprocuniversal} because it avoids computing
a termination argument for $f$ if one of its callees is potentially
non-terminating.

\begin{theorem}
If the entry {\function} of a program is declared terminating,
then the program terminates universally.
If the entry {\function} of a program is declared non-terminating,
then the program never terminates.
\end{theorem}
\begin{proof}
By induction over the acyclic call graph
using Prop.~\ref{prop:ipterm}.
\end{proof}

\subsection{Preconditions for Termination}\label{sec:precond}

Before introducing conditional termination, we have to talk about
preconditions for termination.

If a {\function} terminates conditionally like {\function} $h$ in
Fig.~\ref{fig:example1} $\comptermarg$ (Lemma~\ref{prop:term1})
will not be able to find a satisfying predicate $\rrank$.
However, we would like to know under which
preconditions, i.e. values of \texttt{y} in above example, the
{\function} terminates.

\begin{algorithm}[t]
\KwIn{{\function} $f$ with invariant $\inv$, additional termination
  conditions $\termconditions$
}
\KwOut{precondition $\precond$}
$(\precond,p) \gets (\false,\true)$\;
\KwLet $\varphi = \init(\vxin,\vx) \wedge \inv(\vx)$\;
\While{$\true$}{
    $\psi \gets p \wedge \neg \precond(\vxin) \wedge \varphi$\;
    solve  $\psi$ for $\vxin,\vx,\vxp$\label{line:bootstrap}\;
    \lIf{UNSAT}{\Return{$\precond$}}
    \Else{
        \KwLet $\vxinv$ be a model of $\psi$\;
        \KwLet $\inv' = \compinv(f,\vxin{=}\vxinv)$\label{line:invcand}\;
        \KwLet $\rranktempl = \comptermarg(f,\inv')$\label{line:termcand}\;
        \lIf{$\rranktempl = \true$}{
            $p \gets p \wedge (\vxin\neq \vxinv)$\label{line:blockcand}}
        \Else{
            \KwLet $\theta = \termconditions \wedge \rranktempl$\;
            \KwLet $\precond' = \neg\compprecond(f,\neg\theta)$\label{line:necprecond}\;
            $\precond \gets \precond \vee\precond'$\label{line:addprecond}\;
        }
    }
}
\caption{\label{alg:compPrecond}$\compprecondterm$}
\end{algorithm}

We can state this problem as defined in Def.~\ref{def:precond}.
In Algorithm~\ref{alg:compPrecond} we search for $\precond$, $\inv$, and
$\rrank$ in an interleaved manner.  Note that $\false$ is a trivial
solution for $\precond$; we thus have to aim at finding a good
under-approximation of the maximal solution (weakest precondition) for
$\precond$.

We bootstrap the process by assuming $\precond=\false$ and search for
values of $\vxin$ (Line~\ref{line:bootstrap}).
If such a value $\vxinv$ exists, we can compute an invariant under the
precondition candidate $\vxin=\vxinv$ (Line~\ref{line:invcand}) and 
use Lemma~\ref{prop:term1} to search for the corresponding termination
argument (Line~\ref{line:termcand}).

If we fail to find a termination argument ($\rranktempl=\true$), we block
the precondition candidate (Line~\ref{line:blockcand}) and restart the
bootstrapping process.
Otherwise, the algorithm returns a termination argument $\rranktempl$ that
is valid for the concrete value $\vxinv$ of $\vxin$.
Now we need to find a sufficiently weak $\precond$ for which
$\rranktempl$ guarantees termination.
To this end, we compute an over-approximating precondition for those
inputs for which we cannot guarantee termination ($\neg\theta$ in
Line~\ref{line:necprecond}, which includes additional termination
conditions coming from the backward calling context and preconditions
of {\function} calls, see Sec.~\ref{sec:condterm}).  The negation of
this precondition is an under-approximation of those inputs for which
$f$ terminates.
Finally, we add this negated precondition to our $\precond$
(Line~\ref{line:addprecond}) before we start over the bootstrapping
process to find precondition candidates outside the current precondition
($\neg \precond$) for which we might be able to guarantee termination.

\paragraph{Example.} 
Let us consider again function \texttt{h} in Fig.~\ref{fig:example1}.
This time, we will assume we have the invariant $0 \leq x \leq M$ (with $M:=2^{32}-1$).
We bootstrap by assuming $\precond=\false$ and searching for
values of $y$ satisfying
$\true \wedge \neg\false \wedge x{=}0 \wedge 0\leq x\leq M$. 
One possibility is $y=0$.
We then compute the invariant under the precondition $y=0$ and get
$x=0$. Obviously, we cannot find a termination argument in this case.
Hence, we start over and search for values of $y$ satisfying
$y\neq 0 \wedge \neg\false \wedge x{=}0 \wedge
0{\leq} x{\leq} 10$. This formula is for instance satisfied by $y=1$.
This time we get the invariant $0 {\leq} x {\leq} 10$ and the ranking
function $-x$. Thus, we have to solve

\smallskip
\centerline{
$
\begin{array}{l}
\exists \vec{\precondparam}:
\precondtempl(y, \vec{\precondparam}) \wedge 0{\leq} x{\leq} M \wedge x'{=}x{+}y \wedge x{<}10\\
\qquad \qquad \Rightarrow \neg (-x {>} -x')
\end{array}
$
}

\smallskip
to compute an over-approximating
precondition over the template $\precondtempl$.
In this case,
$\precondtempl(y, \precondparam)$ turns out to be $y = 0$, therefore its
negation $y \neq 0$ is the $\precond$ that we get.
Finally, we have to check for further precondition candidates,
but
$y\neq 0 \wedge \neg(y\neq 0) \wedge x{=}0 \wedge 0{\leq} x{\leq} M$
is obviously UNSAT. Hence, we return the sufficient precondition for
termination $y\neq 0$.

\subsection{Conditional Termination}\label{sec:condterm}

We now extend the formalisation to Algorithm~\ref{alg:interproc},
which additionally requires the computation of under-approximating
calling contexts and sufficient preconditions for termination
(procedure $\compprecondterm$, see Alg.~\ref{alg:compPrecond}).  

First, $\compprecondterm$ computes in line~\ref{line:invcand} an
over-approx\-ima\-ting invariant $\inv^{o}_{fp}$ entailed by the 
candidate precondition.
$\inv^{o}_{fp}$ is computed through Def.~\ref{prop:summaryo}
by conjoining the candidate precondition to the antecedent.
Then, line~\ref{line:termcand} computes the corresponding termination
argument $\rrank_f$ by applying Lemma~\ref{prop:term1} using  
$\inv^o_{fp}$ instead of $\inv^o_{f}$.
Since the termination argument is under-approximating, we are sure that $f$
terminates for this candidate precondition if $\rrank_f\neq\true$.

\rronly{
\begin{remark}
  The available under-approximate information
  $\callctx^u_f \wedge \summaries^u_f \wedge
  \preconditions^u_f$, where \\
$
\begin{array}{@{}rl}
\summaries^u_f =
\bigwedge_{\text{calls }h_j\text{ in }f} & g_{h_j} \Longrightarrow \\
&\summary^u_{h_j}(\vxargin_j,\vxargout_j)
\end{array}
$\\
$
\begin{array}{@{}l}
\text{and }\mathit{Preconds}^u_f =
\bigwedge_{\text{calls }h_j\text{ in }f} g_{h_j} \Rightarrow \precond^u_{h_j}(\vxargin_j)
\end{array}
$ 
could be conjoined with the antecedents in
  Prop.~\ref{prop:summaryo} and Prop.~\ref{prop:term1} in order to
  constrain the search space. However, this is neither necessary for soundness
  nor does it impair soundness, because the same information is used
  in Props.~\ref{prop:precondu} and~\ref{prop:callctxu}.
\end{remark}
}

Then, in line~\ref{line:necprecond} of $\compprecondterm$, we compute
under-approximating (sufficient) preconditions for traces satisfying
the termination argument $\rrank$ via over-approxima\-ting the traces
violating $\rrank$.

Now, we are left to
specify the formulae corresponding to the following functions:
\begin{compactitem}
\item $\compcallctx^u$ (Def.~\ref{prop:callctxu})
\item $\compprecond$ (Def.~\ref{prop:precondu})
\end{compactitem}
We use the superscript ${}^\overtilde{u}$ to indicate 
negations of under-approxi\-ma\-ting information.

\begin{definition}[Line~\ref{line:necprecond} of
  $\compprecondterm$]\label{prop:precondu}
A precondition for termination $\precond^u_f$
in backward calling context $\callctx^u_f$ and with forward
invariants $\inv^o_f$ is  
$\precond^u_f \eqdef \neg
\precond^\overtilde{u}_f$, i.e. the negation of a satisfiability witness $\precond^\overtilde{u}_f$ for: 
$$
\begin{array}{r@{}r@{}l}
\multicolumn{3}{l}{\exists \precond^\overtilde{u}_{h_i},\inv^\overtilde{u}_f,\summary^\overtilde{u}_f: \forall \vxin,\vx,\vxp,\vxpp,\vxout:} \\
\multicolumn{3}{l}{~~\neg\callctx^u_f(\vxin,\vxout) \wedge \inv^o_f(\vx) \wedge \rronly{\summaries^o_f \wedge}}\\
\multicolumn{3}{l}{~~~~\pponly{\summaries^o_f \wedge} \summaries^\overtilde{u}_f \wedge \rronly{\assumptions_f(\vx) \wedge} \assertions_f(\vx) \Longrightarrow} \\
&\big(& \init(\vxin,\vxpp) \wedge \inv^\overtilde{u}_f(\vxpp) \wedge \fout(\vx,\vxout) \\
&&\Longrightarrow \inv^\overtilde{u}_f(\vx) 
\wedge \summary^\overtilde{u}_f(\vxin,\vxout) \wedge \precond^\overtilde{u}_f(\vxin)\big)
\\
\pponly{
\end{array}
$$

$$
\begin{array}{r@{}r@{}l}
}
&~~~\wedge \big(& (\neg\rranktempl_f(\vx,\vxp) \vee \preconditions^\overtilde{u}_f)
\wedge \\
&&\inv^\overtilde{u}_f(\vxp) \wedge \trans(\vx,\vxp)  \Longrightarrow \inv^\overtilde{u}_f(\vx)\big)
\end{array}
$$
\noindent$
\begin{array}{@{}rl}
\text{with }\summaries^\overtilde{u}_f =
\bigwedge_{\text{calls }h_j\text{ in }f} & g_{h_j} \Longrightarrow \\
&\neg\summary^u[h](\vxargin_j,\vxargout_j)
\end{array}$

\noindent$
\begin{array}{@{}r@{}l}
\text{and }\mathit{Preconds}^\overtilde{u}_f = 
\bigvee_{\text{calls }h_j\text{ in }f} & g_{h_j} \Longrightarrow \\
& \neg\precond^u[h](\vxargin_j,\vxargout_j).
\end{array}
$
\end{definition}

This formula is similar to Def.~\ref{prop:summaryo}, but
w.r.t.\ backward calling contexts and summaries, and strengthened by
the (forward) invariants $\inv^o_f$.
We denote the negation of the witnesses found for the summary and the
invariant by
$\summary^u_f \eqdef \neg\summary^\overtilde{u}_f$ and 
$\inv^u_f \eqdef \neg\inv^\overtilde{u}_f$, respectively.

\begin{lemma} \label{lem:precondu}
$\precond^u_f$, $\summary^u_f$ and $\inv^u_f$ are under-approximating.
\end{lemma}

\begin{proof}
We compute an over-approximation of the negation of the precondition
w.r.t.\ the negation of the under-approximating termination argument
and the negation of further under-approximating information (backward
calling context, preconditions of {\function} calls) --- by the
soundness of the synthesis (see Thm.~\ref{thm:synthsound} in
Sec.~\ref{sec:sa}), this over-approximates the non-terminating traces,
and hence under-approximates the terminating ones. Hence, the
precondition is a sufficient precondition for termination.
The term $\neg\rrank_f(\vx,\vxp) \vee
\mathit{Preconds}^\overtilde{u}_f$ characterises non-ter\-mi\-nating
states in the invariants of $f$: for these, either the termination
argument for $f$ is not satisfied or the precondition for termination
of one of the callees does not hold.
\end{proof}

Finally, we have to define how we compute the under-appro\-xi\-ma\-ting
calling contexts:
\begin{definition}[$\compcallctx^u$]\label{prop:callctxu}
  The backward calling context $\callctx^u_{h_i}$
  for {\function} call $h_i$ in {\function} $f$ in backward calling
  context $\callctx^u_f$ and forward invariants $\inv^o_f$ is $\callctx^u_{h_i} \eqdef \neg \callctx^\overtilde{u}_{h_i}$, the 
 negation of a satisfiability witnesses for: 
$$
\begin{array}{r@{}r@{}l}
\multicolumn{3}{l}{\exists \callctx^\overtilde{u}_{h_i},\inv^\overtilde{u}_f: \forall \vxin,\vx,\vxp,\vxargin_i,\vxargout_i,\vxout:} \\
\multicolumn{3}{l}{~~\neg\callctx^u_f(\vxin,\vxout) \wedge \inv^o_f(\vx) \wedge \rronly{\summaries^o_f \wedge}}\\
\multicolumn{3}{l}{~~ \pponly{\summaries^o_f \wedge} \summaries^\overtilde{u}_f \wedge \rronly{\assumptions_f(\vx) \wedge} \assertions_f(\vx) \Longrightarrow} \\
&\big(& \fout(\vx,\vxout) \Longrightarrow \inv^\overtilde{u}_f(\vx)\big) \\
\rronly{
\end{array}
$$
$$
\begin{array}{r@{}r@{}l}
}
&~~~\wedge \big(& \inv^\overtilde{u}_f(\vxp) \wedge \trans(\vx,\vxp) \\
&& \Longrightarrow \inv^\overtilde{u}_f(\vx) \wedge \callctx^\overtilde{u}_{h_i}(\vxargin,\vxargout)\big)
\end{array}
$$
\end{definition}

\begin{lemma}\label{lem:callctxu}
$\callctx^u_{h_i}$ is under-approximating.
\end{lemma}
\begin{proof}
The computation is based on the negation of the under-approximating
calling context of $f$ and the negated under-approximating summaries
for the function calls in $f$. By Thm.~\ref{thm:synthsound}, this
leads to an over-approximation of the negation of the
calling context for $h_i$.
\end{proof}

\begin{theorem}
A {\function} $f$ terminates for all values of $\vxin$ satisfying
$\precond^u_f$.
\end{theorem}
\begin{proof}
By induction over the acyclic call graph
using Lemmae~\ref{lem:precondu} and~\ref{lem:callctxu}.
\end{proof}

\subsection{Context-Sensitive Summaries}

The key idea of interprocedural analysis is to avoid re-analysing
{\function}s that are called multiple times.
For that reason, Algorithm~\ref{alg:interproc} first checks whether
it can re-use already computed information.
For that purpose, summaries are stored as implications $\callctx^o
\Rightarrow \summary^o$.
As the call graph is traversed, the possible calling contexts
$\callctx^o_{h_i}$ for a {\function} $h$ are collected over the call sites~$i$.
$\mathit{NeedToReAnalyze}^o$ (Line~\ref{line:needtoreanalze1} in
Alg.~\ref{alg:interproc}) checks whether the current
calling context $\callctx^o_{h_i}$ is subsumed by calling contexts
$\bigvee_i\callctx^o_{h_i}$ that we have already encountered, and if so,
$\summaries[h]$ is reused; otherwise it needs to be recomputed and
$\mathit{join}$ed conjunctively with previously inferred summaries.
The same considerations apply to invariants, termination arguments
and preconditions.
%

\section{Template-Based Static Analysis}\label{sec:sa}

In this section, we give a brief overview of our synthesis engine,
which serves as a backend for our approach 
(it solves the formulae in Definitions~\ref{prop:callctxo}, \ref{prop:summaryo}, 
\ref{prop:precondu}, and \ref{prop:callctxu} (see Sec.~\ref{sec:interproc})).

Our synthesis engine employs template-based static analysis
to compute ranking functions,
invariants, summaries, and calling contexts, i.e., implementations of
functions $\compinvsummary^o$
and $\compcallctx^o$ from the second-order constraints
defined in Sec.~\ref{sec:interproc}.
To be able to effectively solve second-order
problems, we reduce them to first-order by restricting the space of
solutions to expressions of the form $\templ(\vx,\vec{\invparam})$
where 
\begin{itemize}
  \item $\vec{\invparam}$ are parameters to be instantiated with
    concrete values and $\vx$ are the program variables.
  \item $\templ$ is a template that gives a blueprint for the shape
    of the formulas to be computed.
    Choosing a template is analogous to choosing an abstract domain in
    abstract interpretation.  To allow for a flexible choice, we consider
    \emph{template polyhedra} \cite{SSM05}.
\end{itemize}
We state give here a soundness result:
\begin{theorem}\label{thm:synthsound}
Any satisfiability witness $\vec{d}$ of the reduction of the second order constraint for invariants in Def.~\ref{def:inv} using template~$\templ$
  \[
  \begin{array}{rlr}
  \exists \vec{\invparam},\forall \vxin, \vx, \vxp: & \init(\vxin,\vx) \Longrightarrow \templ(\vx,\vec{\invparam}) 
  \\
  \wedge &\templ(\vx,\vec{\invparam})\wedge\trans(\vx,\vxp)\Longrightarrow \templ(\vxp,\vec{\invparam}) 
  \end{array}
  \]
satisfies $\forall \vx: \inv(\vx) \Longrightarrow \templ(\vx,\vec{\invparam})$,
i.e.\ $\templ(\vx,\vec{\invparam})$ is a sound over-approximating invariant.
Similar soundness results hold true  for summaries and calling contexts. 
\end{theorem}
This ultimately follows from the soundness of abstract interpretation
\cite{CC77}.  Similar approaches have been described, for instance, by
\cite{GS07,GSV08,LAK+14}.  However, these methods consider programs
over mathematical integers.

Ranking functions require specialised synthesis techniques.
To achieve both expressiveness and efficiency, we generate linear lexicographic functions~\cite{BMS05,CSZ13}.
Our ranking-function synthesis approach is similar to the TAN tool~\cite{KSTW10} 
but extends the approach from monolithic to lexicographic ranking functions.
Further, unlike TAN, our synthesis engine is much more versatile and configurable,
e.g., it also produces summaries and invariants.

\pponly{Due to space limitations, w}\rronly{W}e refer to
\rronly{Appendix~\ref{sec:appendix}} \pponly{the extended version
  \cite{extended-version}}, which includes a detailed description of
the synthesis engine, our program encoding, encoding of bit-precise
arithmetic, and tailored second-order solving techniques for the
different constraints that occur in our analysis.  In the following
section, we discuss the implementation.

\section{Implementation}\label{sec:impl}

We have implemented the algorithm in \toolname~ \cite{2LSanon},
a static analysis tool for C programs built on the CPROVER
framework, using MiniSat~2.2.0 as back-end solver.  
Other SAT and SMT
solvers with incremental solving support would also be applicable.
Our approach enables us to use a single solver instance per
{\function} to solve a series of second-order queries as required by
Alg.~\ref{alg:interproc}.
This is essential as our synthesis algorithms make thousands of
solver calls.
Architectural settings (e.g.\ bitwidths) can be provided on the command line.

\pponly{Discussions about technical issues w.r.t. bit-preciseness and 
the computation of intraprocedural termination arguments 
can be found in the extended version \cite{extended-version}.}
\rronly{
\paragraph{Bitvector Width Extension}\label{sec:impl:extension}
As aforementioned, the semantics of C allows integers to wrap around when they
over/underflow.
Let us consider the following example, for which we want to find a
termination argument using Algorithm~\ref{alg:compTerm}:

\noindent{
\footnotesize\textbf{void} f() \{ \textbf{for}(\textbf{unsigned char} x; ; x++); \}}

The ranking function synthesis needs 
to compute a value for template parameter $\rankparam$ such that
$\rankparam(x{-}x')>0$ holds for all $x, x'$ under
transition relation $x'{=}x{+}1$ and computed invariant $\true$
(for details of the algorithm refer to \rronly{Appendix~\ref{sec:intraproc}} \pponly{the extended version \cite{extended-version}}).

\medskip
Thus, assuming that the current value for $\rankparam$ is $-1$, the
constraint to be solved (Algorithm~\ref{alg:compTerm} Line~\ref{line:solveterm})
is $\true \wedge x'{=}x{+}1 \wedge
\neg({-1}(x{-}x')){>}0)$, or $\neg({-1}(x{-}(x{+}1)){>}0)$, for short.
While for mathematical
 integers this is SAT, it is UNSAT for signed bit-vectors due to overflows.
For $x{=}127$, the overflow happens such that $x{+}1{=}{-}128$. Thus,
$127{-}({-}128){>} 0$ becomes $-1 {>} 0$, which makes the constraint
UNSAT, and we would incorrectly conclude that $-x$ is a ranking
function, which does not hold for signed bitvector semantics.  However,
if we extend the bitvector width to $k{=}9$ such that the arithmetic in
the template does not overflow, then
$\neg(-1\cdot((\mathit{signed}_9)127{-}(\mathit{signed}_9)({-}128))>0)$ evaluates
to $255>0$, where $\mathit{signed}_k$ is a cast to a $k$-bit signed
integer.
Now, $x{=}127$ is a witness showing that $-x$ is not a valid ranking function.

For similar reasons, we have to extend the bit-width of $k$-bit unsigned
integers in templates to $(k{+}1)$-bit signed integers to retain soundness.

\paragraph{Optimisations}
Our ranking function synthesis algorithm searches for coefficients
$\vec{\rankparam}$ such that a constraint is UNSAT.
However, this may result in enumerating all the values for
$\vec{\rankparam}$ in the range allowed by its type, which is
inefficient.
In many cases, a ranking function can be found for which
$\rankparam_j \in \{-1,0,1\}$.
In our implementation, we have embedded an improved algorithm 
(Algorithm~\ref{alg:compTerm} in Appendix~\ref{sec:intraproc}).
into an outer refinement loop which iteratively extends the range for
$\vec{\rankparam}$ if a ranking function could not be found.
We start with $\rankparam_j \in \{-1,0,1\}$, then we try
$\rankparam_j \in [-10,10]$ before extending it to the whole range.

\paragraph{Further Bounds}
%
As explained in Algorithm~\ref{alg:compTerm}, we bound the number of
lexicographic components (default 3), because otherwise
Algorithm~\ref{alg:compTerm} does not terminate if there is no number $n$
such that a lexicographic ranking function with $n$ components proves
termination.

Since the domains of $\vx,\vxp$ in Algorithm~\ref{alg:compTerm} and of
$\vxin$ in Algorithm~\ref{alg:compPrecond} might be large, we limit also
the number of iterations (default 20) of the \emph{while} loops in these
algorithms.
In the spirit of bounded model checking, these bounds only restrict
completeness, i.e., there might exist ranking functions or preconditions
which we could have found for larger bounds.
The bounds can be given on the command line.
}
\section{Experiments}\label{sec:exp}

We performed experiments to support the following claims:
\begin{compactenum}
\item Interprocedural termination analysis (IPTA) is faster than
  monolithic termination analysis (MTA).
\item The precision of IPTA is comparable to MTA.
\item {\toolname} outperforms existing termination analysis tools.
\item {\toolname}'s analysis is bit-precise.
\item {\toolname} computes usable preconditions for termination.
\end{compactenum}

We used the \emph{product line} benchmarks of the \cite{svbenchmarks} benchmark repository.
In contrast to other categories, this benchmark set contains programs
with non-trivial procedural structure.
This benchmark set contains 597 programs with 1100 to 5700 lines of code (2705 on average),%
\footnote{Measured using \texttt{cloc} 1.53.}
33 to 136 {\function}s (67 on average), and 4 to 10 loops (5.5 on average).
Of these bencharks, 264 terminate universally, whereas 333 never terminate.

The experiments were run on a
Xeon X5667 at 3\,GHz
running Fedora 20 with 64-bit binaries.
Memory and CPU time were restricted to 16\,GB and 1800 seconds per benchmark, respectively
(using \cite{Rou11}).  Using {\toolname} with interval templates was
sufficient to obtain reasonable precision.

\paragraph{Modular termination analysis is fast}
We compared IPTA with MTA (all {\function}s inlined).
Table~\ref{tab:results2} shows that IPTA times out on 2.3\,\% of the
benchmarks vs.~39.7\,\% for MTA. The geometric mean speed-up of IPTA
w.r.t.~MTA on the benchmarks correctly solved by both approaches is
1.37.

In order to investigate how the 30m timeout affects MTA, 
we randomly selected 10 benchmarks that timed out for 30\,m and re-ran
them: 1 finished in 32\,m, 3 after more than 1h, 6 did not finish
within 2\,h.

\paragraph{Modular termination analysis is precise}
Again, we compare IPTA with MTA.
Table~\ref{tab:results2} shows that IPTA proves 94\,\% of the
terminating benchmarks, whereas only 10\,\% were proven by MTA.
MTA can prove all never-terminating benchmarks including 13 benchmarks
where IPTA times out. MTA times out on the benchmarks that cause 13
additional \emph{potentially non-terminating} outcomes for
IPTA.

\begin{table}[t]
\centering
\begin{tabular}{l|r|rrr@{}l}
     &  \rot{{\toolname} IPTA} & \rot{{\toolname} MTA} & \rot{TAN} & \rot{Ultimate}&  \\
\hline
terminating         & 249 &  26 &  18 &  50& \\
non-terminating & 320 & 333 &   3 & 324&${}^*$ \\
potentially non-term.   &  14 &   1 & 425 &   0& \\
timed out   &  14 & 237 & 150 &  43& \\
errors      &   0 &   0 &   1 & 180& \\
\hline
total run time (h)         &  58.7 & 119.6 & 92.8 & 23.9& 
\end{tabular}~\\[2ex]
\caption{\label{tab:results2} Tool comparison (${}^*$ see text).}
\end{table}

\paragraph{{\toolname} outperforms existing termination analysis tools}
We compared {\toolname} with two termination tools for C programs
from the SV-COMP termination competition,
namely~\cite{tan2014}
and~\cite{Ultimate2015}.

Unfortunately, the tools~\cite{Aprove2014},
\cite{Loopus},
\cite{FuncTion2015},
\cite{HipTNT},
and~\cite{ARMC2011}
have limitations 
regarding the subset of C that they can handle that make them unable
to analyze any of the benchmarks out of the box.
We describe these limitations in \cite{experiments-log}.
Unfortunately, we did not succeed to generate the correct input files
in the intermediate formats required by \cite{MSt2} and
\cite{KiTTeL2015} using the recommended frontends \cite{SLAyer11} and
\cite{llvm2KiTTeL2015}.

TAN~\cite{KSTW10},
and KiTTeL/KoAT~\cite{FKS12} support bit-precise C semantics.
Ultimate uses mathematical integer reasoning but tries to ensure
conformance with bit-vector semantics.  Also, Ultimate uses a semantic
decomposition of the program~\cite{HHP14} to make its analysis
efficient.

Table~\ref{tab:results2} shows lists for each of the tools the 
number of instances solved, timed out or aborted because of an
internal error.
We also give the total run time, which shows that analysis times are
roughly halved by the modular/interprocedural approaches ({\toolname}
IPTA, Ultimate) in comparison with the monolithic approaches
({\toolname} MTA, TAN).
Ultimate spends less time on those benchmarks that it can prove
terminating, however, these are only 19\,\% of the terminating benchmarks
(vs.~94\,\% for {\toolname}).
If Ultimate could solve those 180 benchmarks on which it
fails due to unsupported features of C, we would expect its
performance to be comparable to {\toolname}.

Ultimate and {\toolname} have different capabilities regarding
non-termination.  {\toolname} can show that a program never terminates
for all inputs, whereas Ultimate can show that there exists a potentially
non-terminating execution.  To make the comparison fair, we counted
benchmarks flagged as potentially non-terminating by Ultimate,
but which are actually never-terminating, in the \emph{non-terminating}
category in Table~\ref{tab:results2} (marked ${}^*$).

\paragraph{{\toolname}'s analysis is bit-precise}
We compared {\toolname} with Loopus on a collection of 15 benchmarks
(\texttt{ABC\_ex01.c} to \texttt{ABC\_ex15.c}) taken from the Loopus
benchmark suite \cite{Loopus}.%
While they are short (between 7 and 41 LOC), the
main characteristic of these programs is the fact that they exhibit
different terminating behaviours for mathematical integers and
bit-vectors.
For illustration, \texttt{ABC\_ex08.c} shown Fig.~\ref{fig:example:loopus} 
terminates with mathematical integers, but not 
with machine integers if, for instance, \texttt{m} equals \texttt{INT MAX}. 
Next, we summarise the results of our experiments on these benchmarks
when considering machine integers:
\begin{itemize}
\item only 2 of the programs terminate (\texttt{ABC\_ex08.c} and \texttt{ABC\_ex11.c}), and are correctly 
identified by both {\toolname} and Loopus.
\item for the rest of 13 non-terminating programs, Loopus claims they terminate, 
whereas {\toolname} correctly classifies 9 as potentially non-terminating (including \texttt{ABC\_ex08.c} in Fig.~\ref{fig:example:loopus}) and times out for 4. 
\end{itemize}

\begin{figure}
\begin{center}
\footnotesize
\begin{lstlisting}[basicstyle=\footnotesize]
void ex15(int m, int n, int p, int q) {
  for (int i = n; i >= 1; i = i - 1)
    for (int j = 1; j <= m; j = j + 1)
      for (int k = i; k <= p; k = k + 1)
        for (int l = q; l <= j; l = l + 1)
	;
}
\end{lstlisting} 
\end{center}

\vspace{-.5cm}
\caption{
Example \texttt{ABC\_ex15.c} from the Loopus benchmarks.\label{fig:example:loopus}
}
\vspace{-.5cm}
\end{figure}

\paragraph{{\toolname} computes usable preconditions for termination}
This experiment was performed on benchmarks extracted from Debian
packages and the linear algebra library CLapack.

The quality of preconditions, i.e.~usability or ability to help the
developer to spot problems in the code, is difficult to quantify.
We give several examples where function terminate conditionally.  The
\texttt{abe} package of Debian contains a function, shown in
Fig.~\ref{fig:example:abe}, where increments of the iteration in a
loop are not constant but dynamically depend on the dimensions of an
image data structure.  Here {\toolname} infers the precondition $img
\rightarrow h > 0 \wedge img \rightarrow w > 0$.

\begin{figure}
\begin{center}
\footnotesize
\begin{lstlisting}[basicstyle=\footnotesize]
void createBack(struct SDL_Surface *back_surf)
{
  struct SDL_Rect pos;
  struct SDL_Surface *img = images[img_back]->image;
	
  for(int x=0; !(s>=(*back_surf)->h); s+=img->h) {
    for(int y=0; !(y>=(*back_surf)->w); y+=img->w) {
      pos.x = (signed short int)x;
      pos.y = (signed short int)y;
      SDL_UpperBlit(img, NULL, *back_surf, &pos);
      ...
} } }
\end{lstlisting}
\end{center}

\vspace{-.5cm}
\caption{
Example \texttt{createBack} from Debian package \texttt{abe}.\label{fig:example:abe}
}
\vspace{-.5cm}
\end{figure}

The example in Fig.~\ref{fig:example:busybox} is taken from the
benchmark \texttt{basename} in the \texttt{busybox}-category of SVCOMP
2015, which contains simplified versions of Debian packages.  The
termination of function \texttt{full\_write} depends on the
return value of its callee function \texttt{safe\_write}.
Here {\toolname} infers the calling context $cc>0$, i.e.\ the contract
for the function \texttt{safe\_write}, such that the termination of 
\texttt{full\_write} is guaranteed. Given a proof that 
\texttt{safe\_write} terminates and returns a strictly positive value
regardless of the arguments it is called with, we can conclude 
that \texttt{full\_write} terminates universally.

\begin{figure}
\begin{center}
\begin{lstlisting}[basicstyle=\footnotesize]
signed long int full_write(signed int fd, 
    const void *buf, unsigned long int len, 
    unsigned long int cc) {
  signed long int total = (signed long int)0;
  for( ; !(len == 0ul); 
      len = len - (unsigned long int)cc) {
    cc=safe_write(fd, buf, len);
    if(cc < 0l) {
      if(!(total == 0l))
        return total;
      return cc;
    }
    total = total + cc;
    buf = (const void *)((const char *)buf + cc);
} }
\end{lstlisting}
\end{center}
\vspace{-.5cm}
\caption{
Example from SVCOMP 2015 \texttt{busybox}.\label{fig:example:busybox}
}
\vspace{-.5cm}
\end{figure}

The program in Fig.~\ref{fig:example:clapack} is a code snippet taken from the summation {\function} \texttt{sasum} within
\cite{clapack}, the C version of the popular LAPACK linear algebra library.
The loop in {\function} \texttt{f} does not terminate if $incx=0$.
If $incx>0$ ($incx<0$) the termination argument is that $i$ increases (decreases).
Therefore, $incx \neq 0$ is a termination precondition for \texttt{f}.
\begin{figure}
\begin{center}
\begin{lstlisting}[basicstyle=\footnotesize{}]
int f(int *sx, int n, int incx) {
  int  nincx = n * incx;
  int stemp=0;
  for (int i=0; incx<0 ? i >= nincx: i<= nincx; 
       i+=incx) {
    stemp += sx[i-1];
  }
  return stemp;
}
\end{lstlisting}
\end{center}
\vspace{-.5cm}
\caption{
Non-unit increment from \texttt{CLapack}.\label{fig:example:clapack}
}
\end{figure}

\section{Limitations, Related Works and Future Directions}

Our approach makes significant progress towards
analysing real-world software, advancing the
state-of-the-art of termination analysis of large programs.
Conceptually, we decompose the analysis into a sequence of well-defined second-order
predicate logic formulae with existentially quantified predicates. In
addition to \cite{DBLP:conf/pldi/GrebenshchikovLPR12}, we consider
context-sensitive analysis, under-approximate backwards analysis,
and make the interaction with termination analysis explicit.
Notably, 
these seemingly tedious formulae are actually solved by our generic
template-based synthesis algorithm, making it an efficient alternative
to predicate abstraction.

An important aspect of our analysis is that it is bit-precise. 
As opposed to the synthesis of termination arguments for linear programs over
integers (rationals) \cite{CPR06b,LWY12,BG13,PR04a,HHLP13,BMS05,CSZ13}, 
this subclass of termination analyses is substantially less covered.
While \cite{KSTW10,CKRW10}
present methods based on a reduction to Presburger arithmetic, and a
template-matching approach for predefined classes of ranking functions
based on reduction to SAT- and QBF-solving,
\cite{DBLP:conf/esop/DavidKL15} only compute intraprocedural 
termination arguments.

There are still a number of limitations to be addressed,
all of which connect to open challenges subject
to active research.  While some are orthogonal (e.g., data structures,
strings, refinement) to our interprocedural analysis framework, 
others (recursion, necessary
preconditions) require extensions of it.  
In this section, we discuss related work, as well as, characteristics and limitations
of our analysis, 
and future directions (cost analysis and concurrency).

\paragraph{Dynamically allocated data structures}
We currently ignore heap-allocated data.  
This limitation could be lifted by
using specific abstract domains.  For illustration, let us
consider the following example traversing a singly-linked list.

\begin{lstlisting}[numbers=none]
 List x; while (x != NULL) { x = x->next; }
\end{lstlisting}

Deciding the termination of such a program requires knowledge about
the shape of the data structure pointed by~$x$, namely, the program
only terminates if the list is acyclic. Thus, we would require an
abstract domain capable of capturing such a property and also relate
the shape of the data structure to its length.  Similar to
\cite{CSZ13}, we could use \cite{DBLP:conf/popl/MagillTLT10} in order
to abstract heap-manipulating programs to arithmetic ones.
Another option is 
using an abstract interpretation
based on separation logic formulae which tracks the depths of pieces
of heaps similarly to \cite{DBLP:conf/cav/BerdineCDO06}.

\paragraph{Strings and arrays} 
Similar to dynamic\rronly{ally allocated} data structures, handling strings and
arrays requires specific abstract domains.  String abstractions
that reduce null-terminated strings to integers (indices, length, and
size) are usually sufficient in many practical cases; scenarios where
termination is dependent on the content of arrays are much harder and
would require quantified invariants \cite{DBLP:conf/tacas/McMillan08}.
Note that it is favorable to run a safety checker before the
termination checker.  The latter can assume that assertions 
for buffer overflow checks hold which strengthens invariants and makes
termination proofs easier.

\paragraph{Recursion}
We currently use downward fixed point iterations for computing calling
contexts and invariants that involve summaries (see
Remark~\ref{rem:ipfixpoint}).
This is cheap but gives only imprecise results in the presence of
recursion, which would impair the termination analysis.
We could handle recursions by detecting cycles in the call graph and
switching to an upward iteration scheme in such situations. Moreover,
an adaptation regarding the generation of the ranking function
templates is necessary.
An alternative approach would be to make use 
of the theoretic framework presented in 
\cite{PSW05} for verifying total correctness and liveness properties of while programs
with recursion.

\paragraph{Template refinement}
We currently use interval templates together with heuristics for
selecting the variables that should be taken into consideration.
This is often sufficient in practice, but it does not exploit
the full power of the machinery in place.
While counterexample-guided abstraction refinement (CEGAR) techniques 
are prevalent in predicate abstraction \cite{CGJ+00}, attempts to use them in 
abstract interpretation are rare \cite{RRT08}.
We consider our template-based abstract interpretation that
automatically synthesises abstract transformers more amenable to
refinement techniques than classical abstract interpretations where
abstract transformers are implemented manually.

\paragraph{Sufficient preconditions to termination}
Currently, we compute  
\emph{sufficient} preconditions, i.e. under-approximating preconditions to
termination via computing over-approxima\-ting preconditions to
potential non-termination.
The same concept is used by other works on conditional termination
\cite{CGL+08,BIK12}.
However, they consider only a single {\function} and do not leverage
their results to perform interprocedural analysis on large benchmarks
which adds, in particular, the additional challenge of propagating
under-approximating information up to the entry {\function}
(e.g.\ \cite{GIK13}).
Moreover, by contrast to Cook et al \cite{CGL+08} who use an heuristic
\textsc{Finite}-operator left unspecified for bootstrapping their
preconditions, our bootstrapping is systematic through constraint
solving.

We could compute \emph{necessary} preconditions by computing
over-approximating preconditions to potential termination (and
negating the result).  However, this requires a method for proving
that there exist non-terminating executions, which is a well-explored
topic.
While \cite{GHM+08}
dynamically enumerate lasso-shaped candidate paths for
counterexamples, and then statically prove their feasibility,
\cite{CCF+14} prove nontermination via reduction to safety
proving.  
In order to prove both termination and non-termination, \cite{HLNR10}
compose several program analyses (termination provers for multi-path
loops, non-termination provers for cycles, and 
safety provers).

\paragraph{Cost analysis}
A potential future application for our work is cost and resource analysis. 
Instances of this type of analyses are 
the worst case execution time (WCET) analysis \cite{WEE+08}, as well
as bound and amortised complexity analysis \cite{ADFG10,BEF+14,SZV14}.
The control flow refinement approach \cite{GJK09,CML13}
instruments a program with counters and uses progress invariants to
compute worst case or average case bounds.

\paragraph{Concurrency}
Our current analysis handles single-threaded C programs.
One way of extending the analysis to multi-threaded programs is using
the rely-guarantee technique which is proposed in \cite{Jon83}, and
explored in several works \cite{CPR07,GPR11,PR12} for termination
analysis.  In our setting, the predicates for environment assumptions
can be used in a similar way as invariants and summaries are used in
the analysis of sequential programs.

\section{Conclusions}\label{sec:concl}

While many termination provers mainly target small, hard programs,
the termination analysis of larger code bases has received little attention.
We present an algorithm for \emph{interprocedural termination
  analysis} for non-recursive programs.  To our knowledge, this is the
first paper that describes in full detail the entire machinery
necessary to perform such an analysis.
  Our approach relies on a bit-precise static analysis combining
  SMT solving, template polyhedra and lexicographic,
  linear ranking function templates.
  We provide an implementation of the approach in the static analysis
  tool \toolname, and demonstrate the
  applicability of the approach to programs with thousands of lines of
  code.

\bibliographystyle{ieeetr}
\bibliography{biblio}

\rronly{
\appendix{} \label{sec:appendix}

\subsection{Program Encoding}
%
We encode programs in a representation akin to single-static assignment (SSA).  
We provide a brief review, focusing on the modelling of loops
and {\function} calls.  We continue to use the program in
Fig.~\ref{fig:example1} as example.

In SSA, each assignment to a variable gives rise to a fresh symbol.  For
instance, the initialisation of variable \texttt{x} corresponds to symbol
$x_0$, and the incrementation by $y$ gives rise to symbol $x_2$.
For the return values, additional variables such as $r_f$ are introduced.
In addition, at control-flow join points, the values coming from different
branches get merged into a single $\phi$-variable.  For instance,
$x^{\phi}_1$ is either the initial value $x_0$ or the value of $x$ after
executing the loop, which is denoted as $x^{lb}_3$.  In the case of
branches, the choice is controlled using the condition of the branch.

In the case of loops, the choice between the value at the loop entry point
and the value after the execution of the loop body is made using a
non-deterministic Boolean symbol ${ls}_3$.  That ensures the
SSA remains acyclic and loops are over-approximated.  Moreover, it allows us
to further constrain $x^{lb}_3$ with loop invariants inferred by our
analyses.

In addition to data-flow variables, there are guard variables $g_i$, which
capture the branch conditions from conditionals and loops.  For instance,
the loop condition is $g_2$, and the conditional around the invocation of
$f$ is $g_3$.

To facilitate interprocedural analysis, our SSA contains a placeholder for
{\function} calls, which ensures that {\function} calls are initially
havocked.  It can be constrained using the summaries computed in the course
of the analysis (cf.~Sec.~\ref{sec:prelim}).

Regarding pointers, a may-alias analysis is performed during
translation to SSA form and case splits are introduced accordingly.

\begin{figure}
\begin{center}
	{
	\scalebox{.9}{
	\vspace{-.5cm}
	\scriptsize
		\begin{tabular}{|l|}\hline
		\begin{tabular}{ll}
			$g_{3}$ & $\ssadef z > 0$\\
                        $w_0$ & $\ssadef 0$\\
                        $f(z)$ & \\
			$w_{1}$ & $\ssadef r_h$  \\
			$w^{\phi}_2$ & $ \ssadef g_3 ? w_1 : w_0$\\
                        $r_f$ &  $\ssadef w^{\phi}_2$ \\
                        \end{tabular}
			\\ \hline
			\begin{tabular}{ll}
%

			$g_{0}$ & $\ssadef true$\\			
			$x_{0}$ & $\ssadef 0$\\
			$g_{1}$ & $\ssadef g_{0}$\\
			$x^{\phi}_{1}$ & $\ssadef ({ls}_{3} ? x^{lb}_{3} : x_{0})$\\
			$g_{2}$ & $\ssadef (x^{\phi}_{1} < 10u \wedge g_{1})$\\
			$x_{2}$ & $\ssadef x^{\phi}_{1} + y$\\
			$r_h$ & $\ssadef x^{\phi}_1$\\
                        \end{tabular}
			\\ \hline
		\end{tabular}
	}
	}
	\caption{
		SSA form for example in Fig.~\ref{fig:example1}. \label{fig:ssa:strchr}
	}
\end{center}
\vspace{-.25cm}
\end{figure}

\subsection{Computing Over-Approximating Abstractions}\label{inv}
%
To implement Algorithm~\ref{alg:interproc}, we need to compute
invariants, summaries, and calling contexts, i.e., implementations of
functions $\compinvsummary^o$
and $\compcallctx^o$.  As described in Sec.~\ref{sec:prelim},
invariants and calling contexts can be declaratively expressed in
second-order logic.
To be able to effectively solve such second-order
problems, we reduce them to first-order by restricting the space of
solutions to expressions of the form $\templ(\vx,\vec{\invparam})$
where $\vec{\invparam}$ are parameters to be instantiated with
concrete values and $\vx$ are the program variables.

\paragraph{Template Domain.}
An abstract value $\vec{\invparamv}$ represents the set of all $\vx$ that
satisfy the formula $\templ(\vx,\vec{\invparamv})$ (\emph{concretisation}).
We write $\bot$ for the abstract value denoting the empty set
$\templ(\vx,\bot) \equiv \false$, and $\top$ for the abstract
value denoting the whole domain of $\vx$: $\templ(\vx,\top)
\equiv \true$.

Choosing a template is analogous to choosing an abstract domain in
abstract interpretation.  To allow for a flexible choice, we consider
\emph{template polyhedra} \cite{SSM05}: $\templ =
(\mathbf{A}\vx\leq\vec{\invparam})$ where $\mathbf{A}$ is a matrix
with fixed coefficients.  Polynomial templates subsume intervals,
zones and octagons~\cite{mine06octagons}.  Intervals, for example,
give rise to two constraints per variable $x_i$: $x_i \leq
\invparam_{i1}$ and $-x_i \leq \invparam_{i2}$.  We call the
constraint generated by the $r^{\text{th}}$ row of matrix $\mathbf{A}$
the $r^{\textit{th}}$ \emph{row} of the template.

To encode the context of template constraints,
e.g., inside a loop or a conditional branch,
we use \emph{guarded templates}.
In a guarded template each row $r$ is of the form
$G_r \Longrightarrow T_r$ for the $r^{\text{th}}$ row $T_r$ of the
base template domain (e.g. template polyhedra).
The guards are uniquely defined by the guards
of the SSA variables appearing in $T_r$.
$G_r$ denotes the guard associated to the variables $x$ at the loop
head, and $G_r'$ the guard associated to the variables $x'$ at the end
of the loop body.
A guarded template in terms of the variables at the loop head is of the form:
$
\templ = \bigwedge_{r} G_r \Longrightarrow T_r
$
(respectively $
\templ' = \bigwedge_{r} G'_r \Longrightarrow T'_r
$ if expressed in terms of the variables at the end of the loop body).

\paragraph{Inferring abstractions.}
Fixing a template reduces the second-order search for an invariant to
the first-order search for template parameters:

\[
\begin{array}{rlr}
\exists \vec{\invparam}:
\forall \vxin, \vx, \vxp: & \init(\vxin,\vx) \Longrightarrow \templ(\vx,\vec{\invparam})
\\
\wedge &\templ(\vx,\vec{\invparam}) \wedge\trans(\vx,\vxp)\Longrightarrow \templ'(\vxp,\vec{\invparam}) \ .
\end{array}
\]
By substituting the symbolic parameter $\vec{\invparam}$
by a concrete value $\vec{\invparamv}$,
we see that $\templ(\vx,\vec{\invparamv})$
is an invariant if and only if the following formula is unsatisfiable:
\[
\begin{array}{rlr}
\exists \vxin, \vx, \vxp: & \init(\vxin,\vx) \wedge \neg \templ(\vx,\vec{\invparamv})
\\
\vee & \templ(\vx,\vec{\invparamv}) \wedge\trans(\vx,\vxp)\wedge \neg \templ'(\vxp,\vec{\invparamv}) \ .
\end{array}
\]
As these vectors represent upper bounds on expressions, the most
precise solution is the smallest vector in terms of point-wise
ordering.  We solve this optimisation problem by iteratively calling
an SMT solver. Similar approaches have been described, for instance,
by \cite{GS07,GSV08,LAK+14}.  However, these methods consider programs
over mathematical integers.

Computing overapproximations for calling contexts 
is similar to computing invariants or summaries. 
They only differ in the program variables appearing in the templates: 
$\vxargin,\vxargout$ for calling contexts,
$\vx,\vxp$ for invariants, and $\vxin,\vxout$ for summaries.

\subsection{Termination Analysis For One {\Function}} \label{sec:intraproc}
%
In this section we give details on the algorithm that we use to solve
the formula in Lemma~\ref{prop:term1} (see Sec.~\ref{sec:univterm}).

Monolithic ranking functions are complete, i.e., termination can
always be proven monolithically if a program terminates.  However, in
practice, combinations of linear ranking functions, e.g., linear
lexicographic functions~\cite{BMS05,CSZ13} are preferred, as
monolithic \emph{linear} ranking functions are not expressive enough,
and \emph{non-linear} theories are challenging for existing SMT
solvers, which handle the linear case much more efficiently.

\subsubsection{Lexicographic Ranking Functions}\label{sec:lexi}

\begin{definition}[Lexicographic ranking
  function] \label{def:lexicographic} A lexicographic ranking function
  $R$ for a transition relation $\trans(\vx,\vxp)$ is an
  $n$-tuple of expressions $(\rank_n, \rank_{n-1}, ..., \rank_1)$ such
  that
$$
\begin{array}{rlr}
\multicolumn{3}{l}{\exists \Delta > 0: \forall \vx,\vxp: \trans(\vx,\vxp) \wedge \exists i\in [1,n]:} \\
  & \rank_i(\vx) > 0 & \text{(Bounded)} \\
\wedge & \rank_i(\vx) - \rank_i(\vxp) > \Delta & \text{(Decreasing)} \\
\wedge & \forall j>i: \rank_j(\vx) - \rank_j(\vxp) \ge 0 & \qquad\text{(Unaffecting)}
\end{array}
$$
\end{definition}
Notice that this is a special case of
Definition~\ref{def:ranking_function}.  In particular, the existence
of $\Delta>0$ and the \emph{Bounded} condition guarantee that $>$ is a
well-founded relation.

Before we encode the requirements for lexicographic ranking function into
constraints, we need to adapt it in accordance with the bit-vector
semantics.  Since bit-vectors are bounded, it follows that the
\emph{Bounded} condition is trivially satisfied and therefore can be
omitted.  Moreover, bit-vectors are discrete, hence we can replace the
\emph{Decreasing} condition with $\rank_i(\vx)-\rank_i(\vxp)>0$.  The
following formula, $\lexrank$, holds if and only if $(\rank_n, \rank_{n-1},
..., \rank_1)$ is a lexicographic ranking function with $n$ components over
bit-vectors.
$$\begin{array}{r@{}r@{}l}
\lexrank^n(\vx,\vxp) = &
\bigvee_{i=1}^{n} & (\rank_i(\vx)-\rank_i(\vxp)>0 \; \wedge\\
 & & \bigwedge_{j=i+1}^{n}(\rank_j(\vx)-\rank_j(\vxp)\geq 0))
\end{array}$$
Assume we are given the transition relation $\trans(\vx,\vxp)$ of a
{\function} $f$.  The {\function} $f$ may be composed of several loops, and
each of the loops is associated with a guard $g$ that expresses the
reachability of the loop head (see Sec.~\ref{sec:sa}).
That is, suppose $f$ has $k$ loops, then the lexicographic ranking
function to prove termination of $f$ takes the form:
\begin{center}
$\rrank^{\vec{n}}(\vx,\vxp) = \bigwedge_{i=1}^{k} g_i(\vx) \Longrightarrow \lexrank^{n_i}_i(\vx,\vxp)$
\end{center}

\subsubsection{Synthesising Lexicographic Ranking Functions}\label{sec:complexi}

Ranking techniques for mathematical integers use e.g.\ Farkas' Lemma,
which is not applicable to bitvector operations. We use a synthesis
approach (like the TAN tool~\cite{KSTW10}) and extend it from
monolithic to lexicographic ranking functions.

We consider the class of lexicographic ranking functions generated by
the template where $\rank_i(\vx)$ is the product
$\vec{\rankparam}_i \vx$ with the row vector $\vec{\rankparam}_i$
of template parameters.
We denote the resulting constraints for loop~$i$ as
$\lexranktempl_i^{n_i}(\vx,\vxp,\rankparamvec_i^{n_i})$, where
$\rankparamvec_i^{n_i}$ is the vector
$(\vec{\rankparam}_i^1,\dots,\vec{\rankparam}_i^{n_i})$.
The constraints for the ranking functions of a whole {\function}
$\rranktempl(\vx,\vxp,\vec{\rankparamvec}^\vec{n})$, where
$\vec{\rankparamvec}^\vec{n}$ is the vector
$(\rankparamvec_1^{n_1},\dots,\rankparamvec_k^{n_k})$.

Putting all this together, we obtain the following reduction of ranking
function synthesis to a first-order quantifier elimination problem over
templates:
$$
\exists \vec{\rankparamvec}^\vec{n}: \forall \vx,\vxp:
\inv(\vx)\wedge\trans(\vx,\vxp) \Longrightarrow \rranktempl(\vx,\vxp,\vec{\rankparamvec}^\vec{n})
$$

To complete the lattice of ranking constraints
$\lexranktempl_i^{n_i})$, we
add the special value $\top$ to the domain of
$\rankparamvec_i^{n_i}$.
We define $\lexranktempl_i^{n_i}(\vx,\vxp,\top) \eqdef \true$
indicating that no ranking function has been found for the given
template (``don't know'').
We write $\bot$ for the equivalence class of bottom elements for which
$\lexranktempl_i^{n_i}(\vx,\vxp,\rankparamvec_i^{n_i})$
evaluates to $\false$, meaning that the ranking function has not yet
been computed. For example, $\vec{0}$ is a bottom element.
Note that this intuitively corresponds to the meaning of $\bot$ and
$\top$ as known from invariant inference by abstract interpretation
(see Sec.~\ref{sec:sa}).

\begin{algorithm}[ht]
\KwIn{{\function} $f$ with invariant $\inv$, 
  bound on number of lexicographic components $N$}
\KwOut{ranking constraint $\rranktempl$}
$\vec{n}\gets \vec{1}^k$; $\vec{\rankparamvecv}^\vec{n}\gets
\bot^k$; $\vec{M} \gets \emptyset^k$\label{line:init}\;
\KwLet $\varphi = \inv(\vx)\wedge\trans(\vx,\vxp)$\;
\While{$\true$}{
    \KwLet $\psi = \varphi \wedge \neg \rranktempl(\vx,\vxp,\vec{\rankparamvecv}^\vec{n})$\;
    solve $\psi$ for $\vx,\vxp$\label{line:solveterm}\;
    \lIf{UNSAT}{\Return $\rranktempl(\vx,\vxp,\vec{\rankparamvecv}^\vec{n})$}
    \KwLet $(\vxv,\vxpv)$ be a model of $\psi$\;
    \KwLet $i \in \{i\mid\neg (g_i \Rightarrow \lexranktempl_i^{n_i}(\vxv,\vxpv,\rankparamvecv^{n_i}_i))\}$\;
    $M_i \gets M_i \cup \{(\vxv,\vxpv)\}$\;
    \KwLet $\theta = \bigwedge_{(\vxv,\vxpv)\in M_i} \lexranktempl_i^{n_i}(\vxv,\vxpv,\rankparamvec^{n_i}_i)$\label{line:constrank}\;
    solve $\theta$ for $\rankparamvec_i^{n_i}$\label{line:solverank}\;
    \If{UNSAT}{
      \lIf{$n_i<N$}{
        $n_i \!\gets\! n_i\!+\!1$\label{line:increase};
        $\rankparamvecv_i^{n_i} \!=\! \bot$; $M_i \!=\! \emptyset$
     }
      \lElse{\Return $\rranktempl(\vx,\vxp,\top^k)$}
    }
    \Else{
        \KwLet $m$ be a model of $\theta$\;
        $\rankparamvecv^{n_i}_i \gets m$\label{line:update}\;
    }
}
\caption{\label{alg:compTerm}$\comptermarg$}
\end{algorithm}

We now use the example in Fig.~\ref{fig:example2} to walk through
Algorithm~\ref{alg:compTerm} that we use to solve
Lemma~\ref{prop:term1} (see Sec.~\ref{sec:univterm}).  The left-hand
side of Fig.~\ref{fig:example2} is the C~code and the right side is
its transition relation.  Since the {\function} only has a single
loop, we will omit the guard ($g = \true$).  Also, we assume that we
have obtained the invariant $\inv(x,y) = \true$.
We use Latin letters such as $\x$ to denote variables and Greek
letters such as $\xv$ to denote the values of these variables. 

In each iteration, our algorithm checks the validity of the current
ranking function candidate. If it is not yet a valid ranking function,
the SMT solver returns a counterexample transition.  Then, a new
ranking function candidate that satisfies all previously observed
counterexample transitions is computed.  This process is bound to
terminate because the finiteness of the state space.

\begin{figure}[t]
\vspace*{-3ex}
\parbox{0.2\textwidth}{
\footnotesize
\textbf{int} x=1, y=1; \\
\textbf{while}(x${>}$0) \{ \\
\hspace*{1em}\textbf{if}(y${<}$10) x=nondet(); \\
\hspace*{1em}\textbf{else} x-{-}; \\
\hspace*{1em}\textbf{if}(y${<}$100) y++; \\
\}
}
\parbox{0.3\textwidth}{
\footnotesize
$\begin{array}{r@{}l@{\,}l@{\,}c@{\,}l@{}l}
\multicolumn{5}{@{}l}{\trans((x,y),(x',y')) =}\\
x{>}0 \Rightarrow & \big( & (y{\geq}10 &\Rightarrow& x'{=}x{-}1)\\
& \wedge & (y{<}100 &\Rightarrow& y'{=}y{+}1)\\
& \wedge & (y{\geq}100 &\Rightarrow& y'=y) & \big) \wedge \\
x{\leq}0 \Rightarrow & \big( & \; x'=x &\wedge& y'=y & \big)
\end{array}$
}
\caption{\label{fig:example2}
Example for Alg.~\ref{alg:compTerm} (with simplified $\trans$).
}
\end{figure}

We start from the bottom element (Line~\ref{line:init}) for ranking
functions with a single component (Line~\ref{line:init}) and solve the
corresponding formula $\psi$, which is $\true \wedge
\trans((x,y),(x',$ $y')) \wedge \neg \false$
(Line~\ref{line:solveterm}). $\psi$ is satisfiable with the model
$(1,100,0,100)$ for $(x,y,x',y')$, for instance.
This model entails the constraint $(1\rankparam_x^1 + 100\rankparam_y^1) -
(0\rankparam_x^1 + 100\rankparam_y^1)>0$, i.e.~$\rankparam_x^1 > 0$, in
Line~\ref{line:constrank}, from which we compute values for the template
coefficients $\rankparam_x$ and $\rankparam_y$.
This formula is given to the solver (Line~\ref{line:solverank}) which
reports SAT with the model $(1,0)$ for $(\rankparam_x,\rankparam_y)$,
for example.
We use this model to update the vector of template parameter values
$\rankparamvecv^1_1$ to $(1,0)$ (Line~\ref{line:update}), which
corresponds to the ranking function~$x$.

We then continue with the next loop iteration, and check the current
ranking function (Line~\ref{line:solveterm}).  The formula
$\true \wedge \trans \wedge \neg (x-x'>0)$ is satisfiable by the
model $(1,1,1001,2)$ for $(x,y,x',y')$, for instance.
This model entails the constraint $(\rankparam_x^1 + \rankparam_y^1) -
(1001\rankparam_x^1 + 2\rankparam_y^1)>0$, i.e. 
$1000\rankparam_x^1-2\rankparam_y^1 > 0$) in Line~\ref{line:constrank},
which is conjoined with the constraint $\rankparam_x^1 > 0$ from the
previous iteration.  The solver (Line~\ref{line:solverank}) will tell us
that this is UNSAT.

Since we could not find a ranking function we add another component to the
lexicographic ranking function template (Line~\ref{line:increase}), and try
to solve again (Line~\ref{line:solveterm}) and obtain the model
$(1,99,0,100)$ for $(x,y,x',y')$, for instance.
Then in Line~\ref{line:solverank}, the solver might report that the model
$(0,-1,1,0)$ for
$(\rankparam_x^2,\rankparam_y^2,\rankparam_x^1,\rankparam_y^1)$ satisfies
the constraints.  We use this model to update the ranking function to
$(-y,x)$.

Finally, we check whether there is another witness for $(-y,x)$ not
being a ranking function (Line~\ref{line:solveterm}), but this time
the solver reports the formula to be UNSAT and the algorithm
terminates, returning the ranking function $(-y,x)$.

}

\end{document}